\documentclass[11pt]{article}
\usepackage{amsmath,amssymb,amsthm}
\usepackage[mathscr]{eucal}
\usepackage{rotate,graphics,epsfig}
\usepackage{graphicx}
\usepackage{multirow}
\usepackage[sort]{natbib}
\usepackage{natbib}
\usepackage{graphicx}
\usepackage{tikz}
\usepackage{verbatim}
\usepackage{hyperref}
\usepackage{caption}
\usepackage{subcaption}
\usepackage{float}
\usepackage{booktabs}
\usepackage{algorithm,algorithmicx,algpseudocode}

 \newcommand{\field}[1]{\mathbb{#1}}
 \newcommand{\R}{\field{R}}
 \newcommand{\Exp}{\field{E}}
 \newcommand{\Pro}{\field{P}}

\newcommand{\be}{\begin{eqnarray}}
\newcommand{\ee}{\end{eqnarray}}
\newcommand{\by}{\begin{eqnarray*}}
\newcommand{\ey}{\end{eqnarray*}}
\newcommand{\bn}{\begin{enumerate}}
\newcommand{\en}{\end{enumerate}}
\newcommand{\bi}{\begin{itemize}}
\newcommand{\ei}{\end{itemize}}

\newtheorem{proposition}{Proposition}[section]
\newtheorem{example}{Example}[section]
\newtheorem{theo}{Theorem}[section]
\newtheorem{pr}{Proposition}[section]
\newtheorem{lem}{Lemma}[section]
\newtheorem{co}{Corollary}[section]
\newtheorem{re}{Remark}[section]
\newtheorem{de}{Definition}[section]
\newtheorem{exa}{Example}[section]

\newcommand{\bex}{\begin{exa}}
\newcommand{\eex}{\end{exa}}
\newcommand{\bt}{\begin{theo}}
\newcommand{\et}{\end{theo}}
\newcommand{\bp}{\begin{pr}}
\newcommand{\ep}{\end{pr}}
\newcommand{\bl}{\begin{lem}}
\newcommand{\el}{\end{lem}}
\newcommand{\bc}{\begin{co}}
\newcommand{\ec}{\end{co}}
\newcommand{\br}{\begin{re}}
\newcommand{\er}{\end{re}}
\newcommand{\bd}{\begin{de}}
\newcommand{\ed}{\end{de}}

\numberwithin{equation}{section}

\newcommand{\nb}{\nonumber}

\newbox\tmp
\newdimen\height
\newdimen\dropdist
\def\BX#1{\setbox\tmp=\hbox{$\overline{\scriptstyle #1}$}
              \height=\ht\tmp
              \dropdist=\dp\tmp
              \advance\dropdist by .7pt
              \advance\height by \dp\tmp
              \box\tmp
              \lower \dropdist \hbox{\vrule height
              \height width .25pt\relax}
              \ifnum0=`{\else}\fi
}
\def\bx{\expandafter\BX\expandafter{\ifnum0=`}\fi}

\usepackage[english]{babel}
\usepackage[utf8]{inputenc}
\usepackage{tikz}
\usetikzlibrary{arrows,decorations.pathmorphing,backgrounds,fit,positioning,shapes.symbols,chains}

\usepackage{verbatim}

\usepackage{color}
\usepackage[normalem]{ulem}

\begin{document}

\title{Sample Recycling Method -- A New Approach to Efficient Nested Monte Carlo Simulations}
\author{Runhuan Feng\thanks{Email: rfeng@illinois.edu, Department of Mathematics, University of Illinois at Urbana-Champaign.} , Peng Li\thanks{Email: pengli@nufe.edu.cn, School of Finance, Nanjing University of Finance and Economics.}}
\date{}
\maketitle

\begin{abstract}
Nested stochastic modeling has been on the rise in many fields of the financial industry. Such modeling arises whenever certain components of a stochastic model are stochastically determined by other models. There are at least two main areas of applications, including (1) portfolio risk management in the banking sector and (2) principle-based reserving and capital requirements in the insurance sector. As financial instrument values often change with economic fundamentals, the risk management of a portfolio (outer loop) often requires the assessment of financial positions subject to changes in risk factors in the immediate future. The valuation of financial position (inner loop) is based on projections of cashflows and risk factors into the distant future. The nesting of such stochastic modeling can be computationally challenging.

Most of existing techniques to speed up nested simulations are based on curve fitting. The main idea is to establish a functional relationship between inner loop estimator and risk factors by running a limited set of economic scenarios, and, instead of running inner loop simulations, inner loop estimations are made by feeding other scenarios into the fitted curve. This paper presents a non-conventional approach based on the concept of sample recycling. Its essence is to run inner loop estimation for a small set of outer loop scenarios and to find inner loop estimates under other outer loop scenarios by recycling those known inner loop paths. This new approach can be much more efficient when traditional techniques are difficult to implement in practice.

\bigskip


\end{abstract}
\maketitle
\newpage
\section{Introduction}

Many problems in portfolio risk measurement and financial reporting require nested stochastic modeling. Standard nested Monte Carlo methods can be costly and time consuming to reach a reasonable degree of accuracy. There has been growing demand in the financial industry for methods to speed up the nested simulation procedure.

 In portfolio risk management, nested simulations are applied in a wide variety of risk assessments. Current use of Monte Carlo simulations are typically divided into two stages: outer loops and inner loops. In outer loops, Monte Carlo simulations are performed on all relevant risk factors over a specific risk horizon; the objective is often to calculate some risk measure of a portfolio consisting of multiple financial instruments. In inner loops, those financial instruments are evaluated conditional on risk factors generated from outer scenarios. As mentioned earlier, standard nested Monte Carlo simulations impose heavy computational burden. To tackle this problem, \citet*{gordy2010nested} analyzed the optimal allocation of computational resources between the inner and the outer stage. By minimizing the mean square error of the resultant estimator, they estimated multiple portfolio risk measures such as probability of large losses, Value-at-Risk(VaR), and expected shortfall. Moreover, \citet*{lan2010confidence} constructed confidence intervals based on statistical theory of empirical likelihood and ranking-and-selection method. 
 \citet*{broadie2011efficient} developed a sequential allocation method in the inner stage based on marginal changes of the risk estimator in each scenario. Following their earlier work, \citet*{broadie2015risk} introduced the least square Monte Carlo in the inner stage to estimate the portfolio risk, and \citet*{Hong2017} expanded on the Nadaraya-Watson kernel smoothing method in the inner stage. Recently, \citet*{Giles2019} used the multilevel Monte Carlo method in the nested simulation of risk estimation.

Nested simulations are also commonly used in the insurance literature when financial reporting procedures, such as reserving and capital requirement calculation, are performed under various stochastically determined economic scenarios. \citet*{reynolds2008nested} pointed out that the need of nested stochastic is driven by a number of changes in the regulatory and accounting world and explain the move from stochastic to nested stochastic by a few examples under various accounting standards. A review of various circumstances under which nested simulation arises in financial reporting can be found in \citet*{feng2016nested}. Standard nested Monte Carlo simulations were studied under different accounting requirements, such as the Solvency Capital Requirement(SCR) in Solvency II \citep*{Morgan2006prep,bauer2012on}, reserve and capital with a principle-based approach \citep*{Reynolds2008b}, and the dynamic hedging under Actuarial Guideline (AG) 43 \citep*{feng2016nested}. In the context of AG-43, \citet*{Lifen} replaced the inner stage simulation with PDE numerical approximation in the dynamic hedging. Additionally, universal kriging method and machine learning method improved the efficiency in the stochastic pricing of a large variable annuity portfolio \citep*{gan2013application,gan2015valuation, gan2017efficent}. Most recently, a neural network approach has been used in the SCR of a large portfolio of variable annuity \citep*{Hejazi2017}. A surrogate modeling approach is developed by \citep*{lin_yang} where the functional relationship between input and output of VA valuation models can be approximated by various statistical models. The work is further extended for dynamic hedging of variable annuity portfolio in \citep*{lin_yang_2020}.

All the existing methods to speed up nested simulations can be summarized in three categories: (1) optimal allocation of resources between outer and inner loops \citep*{gordy2010nested,lan2010confidence,broadie2011efficient,Giles2019}, (2) reduction of inner loops through approximation techniques \citep*{broadie2015risk,feng2016nested}, and (3) volume reduction of nested simulation \citep*{Hejazi2017, gan2013application,gan2015valuation}. The second category is more common used in financial reporting due to the ease of implementation. Note that there exist many other fitting methods in the inner stage of the nested simulations, such as the exponential fitting technique \citep*{beylkin2005on} adapted for actuarial applications in \citep*{feng2017analytical}, the multivariate interpolation techniques \citep*{hardy2003investment}, and polynomial approximation.

The new technique proposed in this paper is based on an entirely different strategy. The basic idea is to reduce the number of inner loops by recycling a small set of them for different inner loop estimators. Hence we call this new method sample recycling method (SRM). In contrast with existing methods in the second category, this method completely avoids approximating functional relationship between inner loop estimator and risk factors. Once inner loop paths are generated for an inner loop estimator at some reference point (in state space), we reuse them to compute the estimators at other target points. Estimation with recycled samples requires the distorted weight (density-ratio) based on change of measures. In most well-known Markov models, we can calculate analytical expressions of distorted weights. In general case, one can estimate these distorted weights by non-parametric methods. Through a variety of examples, we will demonstrate the efficiency and applications of both parametric and non-parametric SRMs. 

It was recently brought to our attention that a similar concept to sample recycling was developed in an independent work by \citet*{fengming2017}, which is called the Green simulation method. Their work promotes reusing the output from previous simulation experiments to answer new questions based on simulations. The work of \citet*{fengming2017} and this paper differ in the problem set-up and implementation details. Their work focuses on general stochastic models, whereas this paper frames sample recycling methods in the context of nested stochastic modeling. The Green simulation method uses sample from all previous experiments and do not necessarily use particular sample sets. A mixture likelihood ratio estimator based on samples of all previous experiments is used to estimate quantities with a new input. Hence, in their setting, it is less of an issue to choose appropriate reference samples. In the context of nested simulation, we assume a pre-processed set of sample points. The aim of this paper is to reduce the number of inner loop simulations in a nested stochastic model. Hence the strategy of the sample recycling method is to identify a set of reference outer loop scenarios from which inner loop samples are obtained and to recycle them for the purpose of estimating quantities for other (target) outer loop scenarios. We propose a block method to ensure that sufficient and relevant sample paths are collected to improve the accuracy and efficiency of inner loop estimations. In this method, one reference point is chosen for each block, which effectively control the difference of distributions under reference scenarios and target scenarios. The mixture likelihood ratio method is further studied and extended in the context of tail event estimation in \citet*{dang}. 

The rest of the paper is organized as follows. Section \ref{sec:nested_MC} provides a brief introduction to the standard nested Monte Carlo simulation.
Section \ref{sec:srm} describes the proposed \textit{sample recycling method}, estimator, accuracy, and computational efforts. 
To further illustrate this method, it gives some examples to explain the calculations of inner loops and the estimation of risk measure. 
Section \ref{sec:DR} continues to expand on the sample recycling framework by discussing a data-driven (non-parametric) likelihood estimation method. In both methods, numerical examples are given to compare with the standard nested Monte Carlo simulation and nested simulation via regression. 
Details of mathematical derivations and experiments are presented in the Appendix.

\begin{figure}[h]
\centering
 \includegraphics[width=\textwidth]{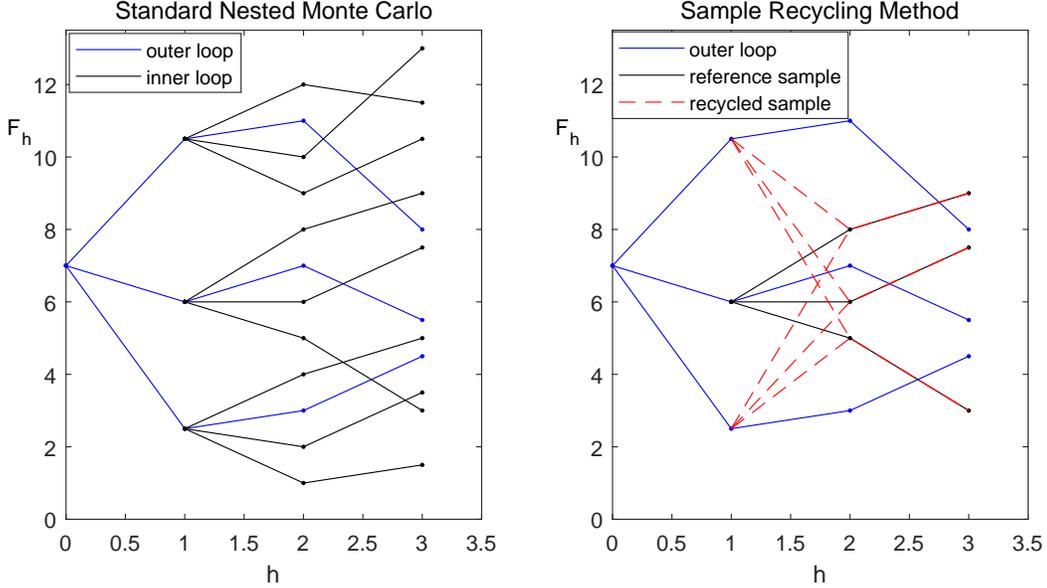}
\caption{Comparison of standard nested MC and sample recycling methods}
\label{fig:srsn}
\end{figure}

\section{Standard nested Monte Carlo method} \label{sec:nested_MC}

In a typical setting of nested simulations, we are interested in the risk measure of a portfolio's loss or gain at some future time $\tau$. This value depends on the evolution of various financial risk factors over the period $[0,\tau]$. Common risk factors may include but are not limited to short-term yield rates, long-term yield rates, equity values, equity volatilities, exchanges rates, etc. Let $\Omega$ be a set of all possible sample paths for risk factors, $\Pro$ be the physical measure under which data are observable in financial markets, $\mathbb{Q}$ be the risk-neutral measure for market consistent valuation. 
Typically, all valuations on portfolio risk management are done under risk-neutral measures. In insurance applications, however, risk measure may be considered under physical measure in financial reporting. As far as the methodology itself is concerned, it does not matter under which measure the application is performed.

\medskip

\noindent {\bf Outer loop estimation}

\medskip

For example, we may consider the risk measure for the valuation of a portfolio
\be \rho=\mathbb{E}[f(L)], \label{riskmeas}\ee
where $L$ is the future loss of the portfolio over the period $[0, \tau]$, and $f$ is a real-valued function such that the expectation exists. Examples of such risk measure may include the probability of a large loss where $f(x)=\mathbf{1}(x\ge c)$, the expected excess loss where $f(x)=(x-c)_+$, and the present value of loss where $f(x)=e^{-r \tau}x$.

In other applications, one may be interested in risk measures such as the Value-at-Risk
\[\mathrm{VaR}_\alpha[f(L)]=\inf\{x\in \R| \Pro(f(L)\le x)>\alpha\},\] or the conditional tail expectation
\[\mathrm{CTE}_\alpha[f(L)]=\Exp\big[f(L)| f(L)>\mathrm{VaR}_\alpha[f(L)]\big],\]
neither of which conforms to the form in \eqref{riskmeas}, which we will focus on for analysis. Nevertheless, it is worthwhile to point out that this sample recycling technique is not restricted to the exact form of \eqref{riskmeas} and can be extended to other risk measures.

\medskip

\noindent {\bf Inner loop estimation}

\medskip

In practice, the computational challenge arises as neither the risk measure in \eqref{riskmeas}  nor the loss random variable $L$ is explicitly expressed by an algebraic formula. Instead, the quantity is estimated in two steps, i.e., a ``nested" setting. The outer layer of the simulation approximates the distribution of the loss $L$ by its empirical distribution as a result of Monte Carlo sampling. In particular, the risk measure can be estimated by the standard statistic given independent and identically distributed samples 
 based on the physical measure,
\be \label{outer}
\frac{1}{n}\sum_{i=1}^{n}f(L_i),
\ee
where $(L_1,L_2,\cdots,L_n)$ is an i.i.d sample of random variable $L$.

Note that the portfolio loss $L_i$ in $i$-th scenario is difficult to compute, as it is usually dependent on paths and cashflows over the period $[\tau, T]$ where $T$ is the specified maturity time. The purpose of inner level simulation is exactly to avoid this difficulty. In practice,  the portfolio loss $L$ is often viewed as a conditional expectation on the information of the risk horizon $[0,\tau]$. Let $\mathcal{F}_{\tau}$ be a field that contains all the information  available to investors at time $\tau$. This conditional expectation can be written as
\be
L=\mathbb{E}^{\mathbb{Q}}[g(Z)|\mathcal{F}_{\tau}],\label{inner}
\ee
where $Z$ 
is a random element of $\mathfrak{R}^d$ describing the performance of portfolio on $[\tau, T]$, and $g(\cdot)$ is a known function from $\mathfrak{R}^d$ to $\mathfrak{R}$.\footnote{Note that $Z$ is defined for simplicity; in general, it can also contain path-dependent situations, for example, the average underlying price in Asian options.} 
To obtain the sample $(L_1,L_2,\ldots,L_n)$, we typically obtain from each outer loop simulation values of the underlying risk factors and generate inner loop sample paths over the period $[\tau, T]$ under the risk-neutral measure. 
Let us denote by $(Z_{i,1}, Z_{i,2},..., Z_{i,m})$ an independent and identically distributed sample of cash flows corresponding to the risk factors for $L_i$. One can think of $\{Z_{i,j}, j=1, \cdots, m \}$ for each fixed $i=1, \cdots, n$ as a set of inner loop paths that emanate from the same initial position determined by the $i$-th outer loop scenario. See the sets of black lines in Figure \ref{fig:srsn}(a) as examples. Then, we can approximate the loss $L_i$ under the $i$-th outer loop scenario by
\be \label{inner2}
\widehat{L}_i:=\frac{1}{m}\sum_{j=1}^{m}g(Z_{i,j}),
\ee
\medskip

\noindent {\bf Standard nested MC estimator}

\medskip

Returning to the outer loop,  the risk measure $\rho$ can be estimated by
 \be \widehat{\rho}_{\text{SN}}=\frac{1}{n}\sum_{i=1}^{n}f(\widehat{L}_i). \label{est_sn_out}
 \ee


As mentioned earlier, there are two issues with the standard nested simulation: computation and accuracy. Previous studies present many methods to accelerate nested simulations, which can be summarized in three categories. (1) {\it Optimal allocation of computation between outer and inner levels.} Such methods are dedicated to decision-making on the number of outer and inner loops given a fixed budget. It is shown that risk estimators with optimal allocation of computational resources presenrs a faster convergence order compared to the uniform allocation schemes \citep*{gordy2010nested,lan2010confidence,broadie2011efficient,Giles2019}. (2) {\it Reduction of inner levels through curve fitting techniques}. The main principle is to find replace the mapping between inner loop estimators and outer loop risk factors. Since the inner-level calculation brings most computational challenge, these methods focus on the approximation of inner loop estimates the proxy functional relationship. A relatively small set of sample is used to estimate the proxy function, which is then used to produce values of inner loop estimator under a wide range of outer loop scenarios. \citep*{broadie2015risk,feng2016nested}. (3) {\it Reduction of the the volume of nested simulation}. The central idea of this category is to strike a balance between computational efficiency and model granularity.  \citep*{gan2013application,gan2015valuation,Hejazi2017}.

 \section{Sample Recycling Method} \label{sec:srm}

Here we introduce a new technique that belongs to the second category: reduction of inner levels. However, the proposed method aims to reduce the number of inner loop simulations based on an entirely different philosophy from curve fitting techniques, such as least square Monte Carlo or pre-processed inner loops. This approach avoids redundant computations in the inner loops by re-sampling a few sets of inner loop paths.

\medskip

\noindent {\bf Inner loop estimation}

\medskip
Bear in mind that the outer loop procedure is kept the same as \eqref{outer} and the proposed method differs from the standard nested simulation and other methods in the inner loop estimation. To consider the new estimator, we typically generate inner loop paths to estimate the loss under a particular scenario. The initial position of risk factors under the particular outer loop scenario is referred to as the {\it reference point}. See the initial position from which the middle set of black lines is generated in Figure \ref{fig:srsn}(b) as an example of the reference point.  Without loss of generality, we consider the reference point to be generated under the $1$-st outer loop scenario.  Recall that the inner loop estimation is carried out for the loss random variable
{\[L_1=\mathbb{E}^{\mathbb{Q}_1}[g(Z)],\]} where $\mathbb{Q}_1$ is the measure under which inner loop sample paths are generated from some initial position determined by the $1$-st outer loop scenario. Then we can determine the inner loop estimator under the $1$-st scenario by
\be \label{bench}
\widehat{L}_1=\frac{1}{m}\sum_{j=1}^{m}g(Z_{1,j}),
\ee
where $(Z_{1,1}, Z_{1,2},\ldots,Z_{1,m})$ are i.i.d samples generated for the random element $Z$ conditioned on the $1$-st outer loop scenario (under measure $\mathbb{Q}_1$). Note that this estimator is the same as the one for standard nested MC method \eqref{inner2}.
 
For simplicity, the $1$-st outer loop scenario is referred to as a reference point and other scenarios as target points. We intend to reuse the inner paths$(Z_{1,1}, Z_{1,2},\ldots,Z_{1,m})$  and evaluations $g(Z_{1,1}), g(Z_{1,2}),...,g(Z_{1,m})$ for the reference point to estimate loss $L_i$ for other target points $i>1$. Denote by $\mathbb{Q}_i$ the probability measure under which the underlying process starts from the  $i$-st outer loop scenario at time $\tau$. The loss for any target point $i$ can be written as
\be
L_i=\mathbb{E}^{\mathbb{Q}_i}[g(Z)]=\mathbb{E}^{\mathbb{Q}_1}[p_{i|1}(Z)g(Z)], \label{qi}
\ee
where $p_{i|1}(\cdot)$ is the Radon-Nikodym derivative of measure $\mathbb{Q}_i$ with respect to $\mathbb{Q}_1$. If the random element $Z$ 
has conditional probability density $p_i(\cdot)$ under $\mathbb{Q}_i$, then the Radon-Nikodym derivative can be given by 
\be p_{i|1}(\cdot)=\frac{d\mathbb{Q}_i}{d\mathbb{Q}_1
}=\frac{p_i(\cdot)}{p_1(\cdot)}.\label{p_rat}\ee
The sample version of the portfolio loss \eqref{qi} can be written as
 \be \label{inner1}
\widetilde{L}_i:=\frac{1}{m}\sum_{j=1}^{m}p_{i|1}(Z_{1,j}) g(Z_{1,j}).
\ee
Under the original measure $\mathbb{Q}_1$ each inner loop sample path $Z_{1,j}$ carries equal weight $1/m$ in \eqref{bench}. In contrast, the evaluation of each inner loop sample path under the measure $\mathbb{Q}_i$ is given a ``distorted" weight in \eqref{inner1}. In general,  we can interpret the weights in the following way. If the recycled path deviates far from the target point, the Radon-Nikodym derivative $p_{i|1}$ gives a small weight, as it is unlikely to observe such a path eminating from the target point. If the recycled path is close to the target point, the derivartive $p_{i|1}$ offers a large weight to reflect its high likelihood.

\medskip

\noindent {\bf Sample recycling estimator }

\medskip
Then the estimation of risk measure $\rho$ by the sample recycling method is given by
\be \label{est_sr_ou}
 \widetilde{\rho}_{\text{SR}}=\frac{1}{n}\sum_{i=1}^{n}f(\widetilde{L}_i).
 \ee

A quick comparison of \eqref{inner2} and \eqref{inner1} shows their differences. Observe that in \eqref{inner2} each estimator under scenario $i$ uses a new sample $(Z_{i,1}, Z_{i,2}, \cdots, Z_{i,m})$, whereas in \eqref{inner1} estimators for all $i=1, \cdots, n$ only use the same sample $(Z_{1,1}, Z_{1,2}, \cdots, Z_{1,m})$. Because all random variables $(Z_{i,1}, Z_{i,2}, \cdots, Z_{i,m})$ are drawn independently under the measure $\mathbb{Q}_1$, all evaluations in \eqref{inner2} are done with equal weight $1/m.$ In contrast, these random variables no longer appear with equal probability under another measure $\mathbb{Q}_i$ for $i>0.$ For this reason, we shall refer to the probability adjustment $p_{i|1}$ as ``distorted" probability.

Now the question is shifted to evaluating the distorted weight $p_{i|1}(\cdot)$. In the discussion above, we assumed for simplicity that $Z$ is $\mathfrak{R}^d$-valued, but the ideas extend to $Z$ taking values in more general sets. Also, we have assumed that $Z$ conditioned on $\omega_i$ has a conditional probability density $p_i(\cdot)$ under $\mathbb{Q}_i$, so that the weight  $p_{i|1}(\cdot)$ is the ratio of two density functions of multidimensional random variable. The following subsection gives a simplified method to determine the weights under Markov models.

\subsection{Distorted weights}\label{sec_weight}

For portfolio management, it is natural to think of $Z$ as the price of underlying assets.
To illustrate the calculation on the distorted weights $p_{i|1}(\cdot)$, we only consider one risk factor and use a Markov process $\{F_t\}_{t\geq0}$ to represent the price of underlying asset.

We consider the discrete path of $\{F_t\}_{t\geq0}$ on the interval $[0, T]$. For  simplicity, let $F_h:=F_{t_h}, h=0,1,2,...,K$ with $t_0=0$ and $t_K=T$, and the risk horizon $\tau=t_k\in [0,T]$.
In this special example, we denote the asset prices under the $i$-th outer loop scenario by $(F^{(i)}_1,F^{(i)}_2,\cdots,F^{(i)}_k)$ for $i=1, \cdots, n$. Under the $i$-th scenario, we can further generate inner loop sample paths $Z_{i,j}=(F^{(i, j)}_{k+1},F^{(i, j)}_{k+2},\cdots,F^{(i, j)}_K)$ for $j=1, \cdots, m$. Suppose that we use the $1$-st scenario as the reference point. We shall recycle sample paths from the reference point, i.e. $(F^{(1, j)}_{k+1},F^{(1, j)}_{k+2},\cdots,F^{(1, j)}_K)$.


We now consider the sample recycling estimator. Observe that $\mathbb{Q}_i$ is the measure under which $F^{(i)}_k$ is realized, i.e. $\mathbb{Q}_i(F_k=F^{(i)}_k)=1.$ In view of \eqref{qi}, we can obtain that
\be
L_i=\mathbb{E}^{\mathbb{Q}_i}[g(F_{k+1},F_{k+2},...,F_K)]=\mathbb{E}^{\mathbb{Q}_1}\left[p_{i|1}(F_{k+1},F_{k+2},...,F_K)g(F_{k+1},F_{k+2},...,F_K)\right],
\ee
where
\begin{equation}\label{multi}
  p_{i|1}(y_{k+1},y_{k+2},...,y_K)= \frac{p_i(y_{k+1},y_{k+2},...,y_K)}{p_1(y_{k+1},y_{k+2},...,y_K)},
\end{equation}
and $p_i(y_{k+1},y_{k+2},...,y_K)$ is the conditional probability density of $(F_{k+1},F_{k+2},...,F_K)$ under $\mathbb{Q}_i$.
The approximation of this weight has high computational cost because it is a ratio of multidimensional density functions. Note that the process of inner simulation is based on the Markov property, indicating that the inner path simulation $(F^{(i, j)}_{k+1},F^{(i, j)}_{k+2},\cdots,F^{(i, j)}_K)$ for $j=1, \cdots, m$  is conditioned on $F_{k}^{(i)}$. This Markov property can also be used in the simulation of the samples of $F_t, t>t_k$. In other words, we can simulate the path of $(F_{k+1}, F_{k+2},\ldots,F_{K})$ through a recursion, for some function $G$,
\be
F_{h+1}=G(F_{h},X_{h+1}), h\geq k,\label{recur}
\ee
which is driven by i.i.d. risk factors $X_{k+1},X_{k+2},...,X_K$. Then the ``distorted'' weight can be reduced to
\be p_{i|1}(y_{k+1},y_{k+2},...,y_K)=\frac{p_i(y_{k+1})f(y_{k+2},...,y_K|y_{k+1}, F_{k}=x_i )}{p_1(y_{k+1})f(y_{k+2},...,y_K|y_{k+1}, F_{k}=x_1)} \nonumber\ee
where $f(\cdot|\cdot)$ is the conditional density function of $(F_{k+2},...,F_K)$ given $(F_{k+1},F_k)$. 
Thanks to the Markov property, it has $f(y_{k+2},...,y_K|y_{k+1}, F_{k}=x_i)=f(y_{k+2},...,y_K|y_{k+1})$, which has no dependence on $F_k$. Hence, the ``distorted" weight can be simplified to
\be
p_{i|1}(y_{k+1},y_{k+2},...,y_K)=\frac{p_i(y_{k+1})}{p_1(y_{k+1})}. \label{dis_wei}
\ee
Hence, according to \eqref{inner1}, the inner loop estimator for the target point can be written as
 \be\label{sim_sta}
\frac{1}{m}\sum_{j=1}^m\frac{p_i(F_{k+1}^{(1,j)})}{p_1(F_{k+1}^{(1,j)})}g(F_{k+1}^{(1,j)},F_{k+2}^{(1,j)},...,F_{K}^{(1,j)}).
\ee
We give the following three examples to further illustrate the simplified weights.

\begin{example}\label{exa1}
In this example, we assume that the price of underlying asset
 $\{F_t\}_{t\geq 0}$  follows a geometric Brownian motion.  The portfolio only has one underlying asset,
and the asset price at the risk horizon $\tau$ (the outer scenario) is driven by, under the real-world measure $\mathbb{P}$
$$dF_t=\mu F_tdt+\sigma F_tdB_t, \ \ \ \ F_0>0,$$
where $\{B_t\}_{t\geq 0}$ is a standard Brownian motion. The loss of portfolio is evaluated under risk-neutral measure $\mathbb{Q}$, under which the asset price is determined by,
\be dF_t=r F_tdt+\sigma F_tdW_t, \ \ \ \ F_0>0,\label{geome}\ee
where $\{W_t\}_{t\geq 0}$ is a standard Brownian motion under risk-neutral measure $\mathbb{Q}$.

There are $n$ scenarios for the asset price before risk horizon $\tau$. We define the corresponding prices at the risk horizon $\tau$ as $x_1,x_2,...,x_n$ where $x_i:=F_{\tau}(\omega_i)$. In each scenario, we can simulate the path of $(F_{k+1}, F_{k+2},\ldots,F_{K})$ through the following recursion
\be
 F^{(i)}_{h+1}=F^{(i)}_h\exp((r-\sigma^2/2)\Delta t+\sigma \sqrt{\Delta t}X_{h+1}), \ \ \ \ h=k,k+1,\ldots,K,\nonumber
\ee
where $X_1,...,X_K$ are independently draw from standard normal distribution with density function $\phi$. This gives the distribution
$$\ln\left(\frac{F^{(i)}_{h+1}}{F^{(i)}_h}\right)\sim N((r-\frac{\sigma^2}{2})\Delta t, \sigma^2\Delta t). $$
We use $x_1$ as the reference point, then the weights \eqref{dis_wei} can be written as
$$p_{i|1}(y)=\frac{\phi\left(\frac{\ln\left(y/x_i\right)-(r-\sigma^2/2)\Delta t}{\sigma\sqrt{\Delta t}}\right)}{\phi\left(\frac{\ln\left(y/x_1\right)-(r-\sigma^2/2)\Delta t}{\sigma\sqrt{\Delta t}}\right)},$$
and the weights can be simplified as
$ p_{i|1}(y) = A {y}^B,$ where coefficients are given by
\begin{align*}
A & = \exp \bigg( \frac{\ln \big( x_i / x_1 \big)}{\sigma^2 \Delta t} \big( -\frac{1}{2} \ln (x_1 x_i) - (r-\frac{1}{2}\sigma^2)\Delta t \big) \bigg), \\
B & = \frac{\ln \big( x_i / x_1 \big)}{\sigma^2 \Delta t}.
\end{align*}
If we insert parameters $x_i = x_1$, i.e., using oneself as a reference, $p_{i|1}(y)=1$ for any $y$, as expected.
\end{example}

\begin{example}
Suppose we have a portfolio exposed to interest rate risk, and let the rate follow a Vasicek model\citep*{Vasicek1977} under a risk neutral measure,
$$ dF_t = \kappa (\theta - r_t) dt + \sigma dW_t, $$
where constants $\kappa, \theta, \sigma$ denote the speed of reversion, the long-term mean level, and the instantaneous volatility respectively. Here $\{W_t\}_{t>0}$ is a pure Brownian motion under the risk neutral measure. Given the risk horizon $\tau$ and the outer scenarios $x_1,x_2,...,x_n$ where $x_i := F_{\tau}(\omega_i)$,  we can simulate the path of $F_t$ on the interval $[\tau, T)$ in each scenario with the following recursion\citep*{glasserman2003monte}
$$F^{(i)}_{h+1}=e^{-\kappa\Delta t} F^{(i)}_h+\theta(1-e^{-\kappa\Delta t})+\sigma\sqrt{\frac{1}{2\kappa}(1-e^{-2\kappa\Delta t})}X_{h+1},$$
where $X_1,...X_{K}$ are independent draws from a standard normal distribution.
Similarly, we can get the weight as follows
$$p_{i|1}(y)=\phi\left(\frac{y-(e^{-\kappa\Delta t} x_i+\theta(1-e^{-\kappa\Delta t}))}{\sigma\sqrt{\frac{1}{2\kappa}(1-e^{-2\kappa\Delta t})}}\right)\Bigg/\phi\left(\frac{y-(e^{-\kappa\Delta t} x_1+\theta(1-e^{-\kappa\Delta t}))}{\sigma\sqrt{\frac{1}{2\kappa}(1-e^{-2\kappa\Delta t})}}\right),$$
which can be simplified to
$ p_{i|1}(y) = A \exp(B y), $ where coefficients are given by
\begin{align*}
A & = \exp \bigg( - \frac{\kappa e^{-\kappa \Delta t}(x_i - x_1) \big( e^{-\kappa \Delta t}(x_i + x_1) + 2\theta (1 - e^{-\kappa \Delta t})\big) }{\sigma^2(1-e^{-2\kappa \Delta t})} \bigg), \\
B & = \frac{2\kappa (x_i - x_1) e^{-\kappa \Delta t}}{\sigma^2(1-e^{-2\kappa \Delta t})}.
\end{align*}
\end{example}
\begin{example}
Assume that the equity return process is modeled by a two-state regime switching log-normal model \citep*{Hardy2001} with parameters $\Theta=\{\mu_1, \sigma_1, \mu_2, \sigma_2, p_{12}, p_{21}\}$. In such a model, the equity process switches between two regimes with low and high volatilities.  Let $s_k:=s_{t_k}$  denote the regime at time $t_k$ and $F_k:=F_{t_k}$ be the  equity return at time $t_k$. The two regimes are represented by $1$ and $2$, i.e. $s_k\in \{1,2\}.$ There are two risk factors in this model, which are modeled by the bivariate process $\{(F_k, s_k), k=1, 2, \cdots \}$. 
The equity return is log-normally distributed, i.e.
 $$ \ln \frac{F_{k+1}}{F_k}\Bigg|_{s_{k+1}}\sim N(\mu_{s_{k+1}}\Delta t, \sigma_{s_{k+1}}^2\Delta t).$$ The transition probability from regime $m$
 to $l$ is given by $p_{ml}=\mathbb{P}(s_{h+1}=l|s_h=m), m,l=1,2$.
Given the risk horizon $\tau$ and the outer scenarios $x_1,x_2,...,x_n$ where $x_i := F_{k}^{(i)}$. We need to simulate the path of $(F_h^{(i)}, s_h^{(i)})$ for $h>k$, which is determined by
\be
 F^{(i)}_{h+1}=F^{(i)}_h\exp(\mu_{s_h}\Delta t+\sigma_{s_h}\sqrt{\Delta t} X_{h+1}), \ \ \ \ h>k,\nonumber
\ee
where $X_1, \cdots, X_{K}$ are independent draws from a standard normal distribution and $s_h$ is the regime applying in the interval $[t_h, t_{h+1})$. The regime $s_{h+1}$ is simulated by a uniform random variable, and is determined by $s_h$ and the transition probability. Then the weight is based on the regime applying in the interval $[t_k, t_{k+1})$ in each scenario. 
Let $q$ represent the density function of $(F_{k+1},s_{k+1})$ conditioned on $(F_k,s_k)$, then we have
$$q(y,s_{k+1}|x,s_k)=\mathbb{P}(s_{k+1}|s_{k})f(y|s_{k+1}, x),$$
where $$f(y|s_{k+1}, x)=\phi\left(\frac{\log\left(\frac{y}{x}\right)-\mu_{s_{k+1}}\Delta t}{\sigma_{s_{k+1}}\sqrt{\Delta t}}\right),$$
and $\phi$ is the standard normal probability density function. In such a model, the distorted weight is given by $p_{(i,m)|(1,l)}$ where $l,m$ are the states of reference point and target point, respectively.
Define $q_{i,m}(y,s):=q(y,s|F_k=x_i,s_k=m)$ and $f_i(y|s)=f(y|s, F_k=x_i)$. Therefore, for $m,l=1,2$ the weights can be written as
$$p_{(i,m)|(1,l)}(y, s)=\frac{q_{i,m}(y,s)}{q_{1,l}(y,s)}, $$
which can be simplified to
\[p_{(i,m)|(1,l)}=\frac{p_{ms}}{p_{ls}}\frac{f_i(y|s)}{f_1(y|s)}.\]

\end{example}

\subsection{Analysis of Estimators}

\subsubsection{Bias and variance}
In this subsection, we analyze the bias and variance of estimator  $\widetilde{L}_i$ under $\mathbb{Q}_i$ and the convergence of estimator $\widetilde{\rho}_{SN}$. The error analysis of statistic $\widetilde{L}_i$ is similar to the importance sampling method  \citep*{Hesterberg1995,Skare2003}, and the convergence of $\widetilde{\rho}_{SN}$ is an extension of the work on the standard nested Monte Carlo \citep*{Rainforth2018a}.
\begin{proposition}\label{thm:bias_var}
The asymptotic bias and variance of $\widetilde{L}_i$ are given by
\be
&& \text{Bias} (\widetilde{L}_i)=0,\nonumber\\
&&\text{Var}^{\mathbb{Q}_1}(\widetilde{L}_i)=O\left(\frac1m\right),\qquad \mbox{ as } m \rightarrow \infty.
\ee
\end{proposition}

\begin{proof}
It follows from \eqref{qi} and \eqref{inner1} that the estimator $\widetilde{L}_i$ is unbiased. Since $(Z_{1,1}, Z_{1,2},...,Z_{1,m})$ 
is a sample of i.i.d. random variables generated from the random element $Z$ conditioned on the $1$-st outer loop scenario. The variance of $\widetilde{L}_i$ under $\mathbb{Q}$ can be written as
\be
\text{Var}^{\mathbb{Q}_1}\left(\widetilde{L}_i \right)&=&\frac1m \text{Var}^{\mathbb{Q}_1}[p_{i|1}(Z_1)g(Z_1)]
=\frac1m\Big(\mathbb{E}^{\mathbb{Q}_i}[p_{i|1}(Z_1)g^2(Z_1)]-L_i^2\Big).\nonumber \ \ \ \ \ 
\ee
\end{proof}
Here we provide some comparison of the variances of $\hat{\rho}_{\mathrm{SN}}$ and $\hat{\rho}_{\mathrm{SR}}.$ In particular, we focus on the special case that $f(x)=x$ in \eqref{riskmeas} and $\rho=\mathbb{E}(L)$. For brevity, we denote for $i=1,2,$
\by
&&A_l:=\mathbb{E}\left[\left(p_{i|1}(Z_{1,1})g(Z_{1,1})\right)^l\right], \qquad B_l:=E[\left(g(Z_{1,1})\right)^l],  \\ &&C:=\mathbb{E}\left[p_{i|1}(Z_{1,1})\left(g(Z_{1,1})\right)^2\right], \qquad D:=\mathbb{E}\left[p_{i|1}(Z_{1,1})p_{j|1}(Z_{1,1})\left(g(Z_{1,1})\right)^2\right],
\ey
\begin{proposition} \label{prop_variance}
The variances can be written as
\be
&&\text{Var}(\widehat{\rho}_\mathrm{SN})=O\left(\frac{1}{mn}\right),\\
&&\text{Var}(\widetilde{\rho}_\mathrm{SR})= O\left(\frac{1}{m}\right).
\ee
\end{proposition}
\begin{proof}
Since $\widehat{L}_i, i=1,2,...,n$ are i.i.d. estimators and $Z_{i,1},Z_{i,2},...,Z_{i,m}$ are i.i.d. random variables, it follows from \eqref{inner2} and \eqref{est_sn_out} that
$$\text{Var}(\widehat{\rho}_{\text{SN}})=\frac{1}{n}\text{Var}\left(\widehat{L}_1\right) = \frac{1}{mn} \text{Var}(g(Z_{1,1}))=\frac{1}{mn}(B_2-B_1^2).$$
In view of \eqref{inner1}, and \eqref{est_sr_ou}, we can write
\be
\widetilde{\rho}_\mathrm{SR}=\frac{1}{n}\sum_{i=1}^n\widetilde{L}_i= \frac{1}{n}\sum_{i=1}^n\frac{1}{m}\sum_{j=1}^mp_{i|1}(Z_{1,j})g(Z_{1,j})=\frac{1}{m}\sum_{j=1}^m \left(\frac{1}{n}\sum_{i=1}^n p_{i|1}(Z_{1,j})g(Z_{1,j})\right)\nonumber.
\ee
Therefore,
\be
\text{Var}(\widetilde{\rho}_{\mathrm{SR}})&=&\frac{1}{m}\text{Var}\left(\frac{1}{n}\sum_{i=1}^n p_{i|1}(Z_{1,1})g(Z_{1,1})\right)=\frac{1}{mn^2}\text{Var}\left[g(Z_{1,1})\left(1+\sum_{i=2}^n p_{i|1}(Z_{1,1})\right)\right]\label{equa_smr}\\
&=&\frac{1}{mn^2}\left(\mathbb{E}\left[\left(g(Z_{1,1})\left(1+\sum_{i=2}^np_{i|1}(Z_{1,1})\right)\right)^2\right]-\left(\mathbb{E}\left[g(Z_{1,1})\left(1+\sum_{i=2}^n p_{i|1}(Z_{1,1})\right)\right]\right)^2\right)\nb\\
&=& \frac{1}{mn^2}\left[B_2+(n-1)A_2+2(n-1)C+(n^2-3n+2)D-(B_1+(n-1)A_1)^2\right].\nb
\ee
The asymptotics follow immediately from the results above.
\end{proof}

\begin{figure}[h]
\includegraphics[width=1.1\textwidth]{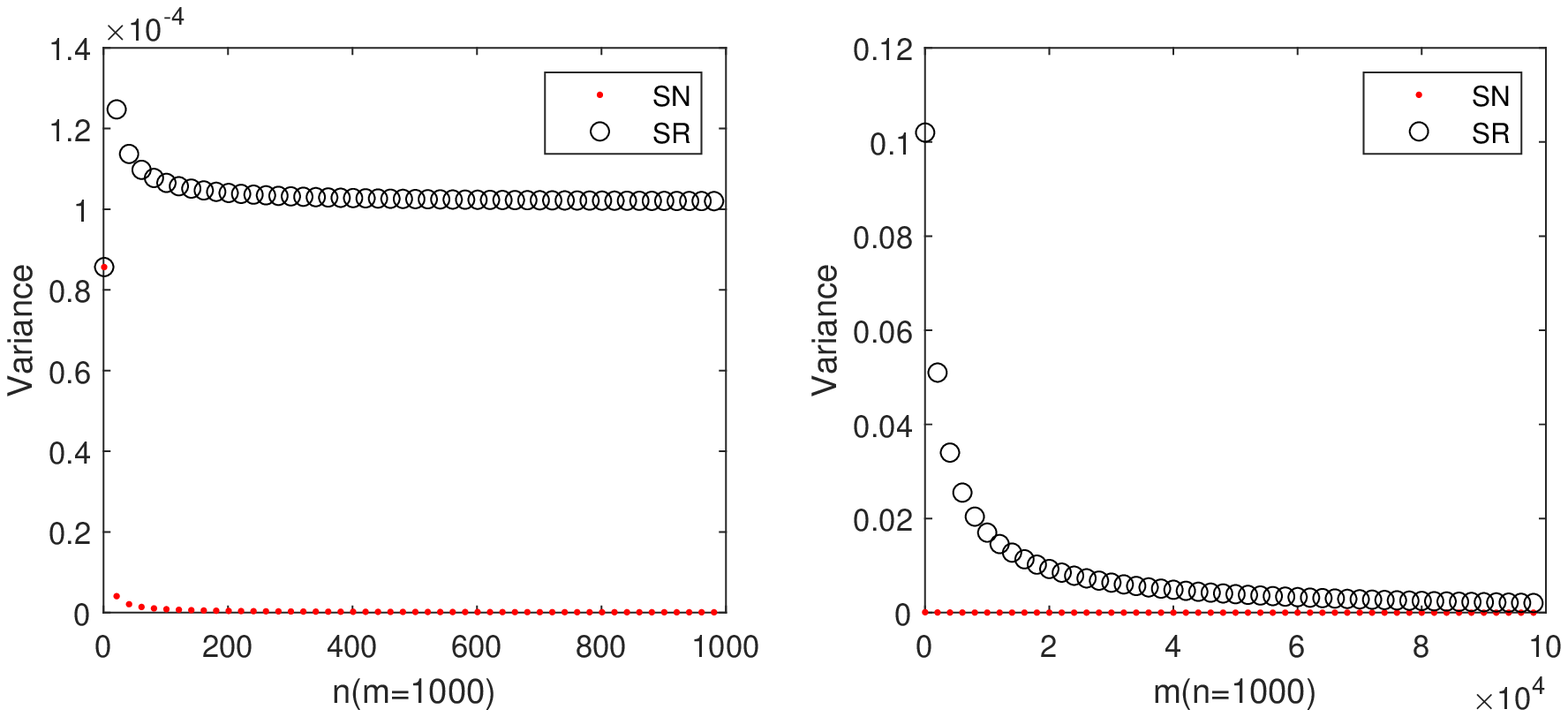}
\caption{Variances of $\widehat{\rho}_{\mathrm{SN}}$ and $\widetilde{\rho}_{\mathrm{SR}}$.
}\label{variance}
\end{figure}

To illustrate convergence rates, we consider an example where $X$ is uniformly distributed on $[-1,1]$ and $Z$ has a standard normal distribution in \eqref{inner}. In such a case, we can calculate the exact loss $L=\mathbb{E}[g(Z)|X]=\mathbb{E}[\sqrt{2/\pi}\exp (-2(Z-X))|X].$ Details of the calculation are left in Appendix \ref{sec:app}. The left panel in Figure \ref{variance} shows the changes in variances of $\widehat{\rho}_{\mathrm{SN}}$ and $\widetilde{\rho}_{\mathrm{SR}}$ with the increasing number of outer loops $n$ and the fixed number of inner loops $m=1,000$. When $n=1$, both estimators are precisely the same as there is only one set of inner loop paths and hence they have the same variance, i.e. $1/m\text{Var}(g(Z_{1,1}))$. When $n=2$, the jump in the variance of SRM estimator is due to the presence of error from using the inner loop sample of a reference point for the target point. As $n$ increases, while  the variance of $\widehat{\rho}_{\mathrm{SN}}$ decreases, it does not diminish as quickly as that of $\widetilde{\rho}_{\mathrm{SR}}.$ This is because all target points on outer loop scenarios use exactly the same set of inner loop paths from the reference point. All portfolio loss estimators $\tilde{L}_i$'s are driven by the same source of randomness $(Z_{1,1}, \cdots, Z_{1,m}).$ Therefore, they tend to overestimate or underestimate all in the same direction and the sample errors in $\widehat{L}_i$ do not offset each other. In contrast, each estimate of $\widehat{L}_i$ is based on an independent sample of $(Z_{i,1}, \cdots, Z_{i,m})$ and hence the sample errors in $\widehat{\rho}_{\mathrm{SN}}$ average out. The right panel of Figure \ref{variance} shows the convergence of variances of $\widehat{\rho}_{\mathrm{SN}}$ and $\widetilde{\rho}_{\mathrm{SR}}$ with an increasing number of inner loop paths $m$ and a fixed number of outer loop scenarios $n=1,000.$ In such an experiment, the increased inner loop sample size significantly improves the accuracy of estimation involving the reference point and hence in turn reduces the error in the estimation of other target points.  The value of the difference $\mathrm{Var}(\widehat{\rho}_{\mathrm{SN}})-\mathrm{Var}(\widetilde{\rho}_{\mathrm{SR}})$ converges to the constant $1.0150\times 10^{-4}$ which is given by $(D-A_1^2)/m$. This numerical example confirms the observation earlier that the standard nested Monte Carlo estimator tends to converge faster than the sample recycling method. The real purpose of the sample recycling method is to give up some accuracy in exchange for high efficiency for any fixed computational budget. The comparison of computational effort is discussed in the next subsection.


\subsubsection{Computational efforts}
 To compare the computational effort, we should first look at algorithms of both standard nested Monte Carlo and sample recycling methods.
\begin{algorithm}[H] 
\label{algo1}
\caption{Estimate risk measure $\rho$ using $\widehat{\rho}_{\mathrm{SN}}$}
\begin{algorithmic}
\State Generate $n$ outer scenarios 
\For {$\ i = 1 \  \text{to}  \ n$}
\State Conditioned on scenario $\omega_i$, generate $m$ i.i.d. inner pathes $Z_{i,1}, Z_{i,2},\ldots,Z_{i,m}$  
\State Calculate the sequence $g(Z_{i,1}), g(Z_{i,2}),...,g(Z_{i,m})$
\State $\widehat{L}_i \leftarrow (1/m) \sum_{j=1}^m g(Z_{i,j})$
\EndFor
\State $\widehat{\rho}_{\mathrm{SN}} \leftarrow(1/n)\sum_{i=1}^n f(\widehat{L}_i)$

\end{algorithmic}
\end{algorithm}
\begin{algorithm}
\label{algo2}
\caption{Estimate risk measure $\rho$ using $\widetilde{\rho}_{\mathrm{SR}}$ }
\begin{algorithmic}
\State Generate $n$ outer scenarios 
\State Conditioned on the $1$-st scenario, generate $m$ i.i.d. inner pathes $Z_{1,1}, Z_{1,2},...,Z_{1,m}$
\State Calculate the sequence $g(Z_{1,1}), g(Z_{1,2}),...,g(Z_{1,m})$
\For {$\ i = 1 \  \text{to}  \ n$}
\State Calculate the sequence $p_{i|1}(Z_{1,1}), p_{i|1}(Z_{1,2}),...,p_{i|1}(Z_{1,m}) $
\State $\widetilde{L}_i\leftarrow \frac{1}{m}\sum_{j=1}^{m}p_{i|1}(Z_{1,j}) g(Z_{1,j})$
\EndFor
\State $\widetilde{\rho}_{\mathrm{SR}} \leftarrow(1/n)\sum_{i=1}^n f(\widetilde{L}_i)$
\end{algorithmic}
\end{algorithm}
According to these algorithms, the estimation of $\widehat{\rho}_{\mathrm{SN}}$ requires generating a total of $nm$ inner paths and evaluating the $g(Z_{i,j})$ for a total of $nm$ times, while the estimation of $\widetilde{\rho}_{\mathrm{SR}}$ uses only $m$ inner paths, the computation of $g(Z_{1,j})$ for $m$ times and that of $p_{i|1}(Z_{1,j})$ for $(n-1)m$ times. In other words, we can measure the computational efforts with the following units:
\begin{itemize}
  \item$ \gamma:=  \text{simulation time of each inner path} \ Z_{i,j}+\text{calculation time of each} \ g(Z_{i,j})$,
  \item $\delta:= \text{calculation time of each} \ p_{i|1}(Z_{1,j})$.
\end{itemize}
We use CE to denote the computational effort of each method. Then the computational efforts required by the two methods are given by
\be
&& \text{CE}_{\mathrm{SN}}=nm\gamma \label{cesn}\\
&&  \text{CE}_{\mathrm{SR}}=m\gamma+(n-1)m\delta. \label{cesr}
\ee
The main computational difference depends on the sizes of $\gamma $ and $\delta$. It is clear that when $\gamma=\delta$ the two methods require exactly the same amount of computational resources. Note, however, that $\gamma$ includes the computation of each inner loop and cash flow projection.  If the financial instrument is path-dependent, then such a calculation can be very time-consuming. While the value of $\delta$ is determined by a likelihood, the Radon-Nikodym derivative is not path-dependent in Markov models as shown in \eqref{dis_wei}. The real advantage of sample recylcing method is only shown when $\gamma$ far exceeds $\delta,$ which is often the case with long-term products and very sophisticated evaluation of cash flows.

\subsection{Extension to multiple reference points} \label{sec:multi}

It follows from Theorem \ref{prop_variance} that for a fixed number of outer loop scenarios the sample recycling estimator achieves the same rate of convergence, $O(1/m)$ as the standard Monte Carlo estimate. Nonetheless, the main advantage of this method is to enhance efficiency by reducing computational efforts. In order to improve the accuracy of this method, one can introduce multiple reference points for variance reduction. For example, consider $b$ reference points $x_1, x_2, \ldots, x_b$ and we want to estimate the portfolio loss for the target point $x_i$ ($i > b$). We can take a weighted average of estimates based on individual reference points given in \eqref{inner1},
$$ \widetilde{L}_i :=  \sum_{k=1}^{b} w_{ik} \bigg( \frac{1}{m}\sum_{j=1}^{m}p_{i|k}(Z_{k,j}) g(Z_{k,j}) \bigg), $$
where the weights shall satisfy $ w_{ik} \geq 0$ and $\sum_{k=1}^b w_{ik} = 1$ for each target point $i =b+1, \cdots, n$.

A simple approach is to use equal weights, i.e. $ w_{ik} =1/b $ for $k=1, \cdots, b$, which corresponds to the simple average of $\widetilde{L}_i$ estimated using each reference point. Since samples generated for reference points are mutually independent, an advantage of this approach is the reduction of variance of $\widetilde{L}_i$ due to the increase of sample size to $mb$,
$$ \text{Var}\bigg[ \sum_{k=1}^{b} \frac{1}{b} \bigg( \frac{1}{m}\sum_{j=1}^{m}p_{i|k}(Z_{k,j}) g(Z_{k,j}) \bigg) \bigg] = \frac{1}{mb} \text{Var}\big[ p_{i|k}(Z_{k,j}) g(Z_{k,j}) \big]. $$

Another approach is to apply a proximity rule. We can break the entire range of scenarios into a number of blocks and select one reference point in each block. Then we generate a set of inner risk paths for each reference point. Inner loop estimation for other target points in each block uses only the reference point in that block, i.e., $w_{ik}= 1$ if $i$ is in the block with $k$ and $w_{ik}=0$ otherwise. This consideration is inspired from potential higher variance due to reference points being far from target, which is reflected in the numerator terms in $\text{Var}(\widetilde{L}_i)$ in Theorem \ref{thm:bias_var},
$$ \text{Var}(\widetilde{L}_i) = \frac1{m}\mathbb{E}[p_{i|k}(Z_k)g^2(Z_k)]-(L_i)^2 \qquad \text{v.s.} \qquad \text{Var}[\widehat{L}_i] = \frac1m \mathbb{E}[g^2(Z_i)]-(L_i)^2. $$
If reference point $k$ is properly chosen for each $i$, then we can achieve a reduction in variance.

In the following examples, we consider a single risk factor for simplicity and use the absolute difference as a metric to assign target scenarios into blocks. In higher dimensional or more complicated cases, one can define more suitable distance metrics on the sample space of risk factor $F$ for the assignment of reference points. There are two common methods for block partitioning. (1) Equidistant partition: keep the same distance between boundary points in each block; (2) Quantile partition: use order statisics or empirical quantiles to allocate $F_k$ into blocks, each of which contains the same number of points. For each block, we shall choose one reference point, for example, the midpoint or a boundary point.
\subsection{Numerical examples}

As a trade-off between sampling variance and computational effort, we observe that the sample recycling method tends to reduce computational effort at the expense of increased variance. We offer a number of examples where the inner simulation and the evaluation of $g(Z_{i,j})$ can be computationally much more challenging than that of $p_{i|1}()$.

We assume that all of the underlying asset prices $\{F_t\}_{t\geq 0}$ follow geometric Brownian motion processes and that asset prices at the risk horizon $\tau$ (outer scenarios) are evaluated under a real-world measure $\mathbb{P}$. While in theory we can use a single reference point to estimate portfolio losses for all other target points, our experiments show that more reference points can significantly improve accuracy. There are many methods to determine the reference points. For example, the reference points can be chosen equidistantly. In each trial, we sort the samples $(x_k=F_{\tau}^{(k)}, k=1,2,\ldots,n) $ and calculate the difference between the maximum and the minimum. Let $s$ be the block number, then samples can be divided equidistantly into $s$ intervals, and each interval has same range. In each block,  the intermediate point or endpoints can be chosen as the reference points. In the 
following numerical example, we implement the estimation by dividing blocks. 


\begin{example}\label{ex_bar}
Consider an asset with initial price $F_0 = 100$, real-world drift $\mu = 8 \% $, and instantaneous volatility $\sigma_0 = \sigma_1 = 20\%$. Let the risk-free continuously compounding interest rate be $r=3\%$. Construct a portfolio of three partial-time barrier options that can only be knocked in or out on the interval $[\tau, T]$ where risk horizon is $\tau=1/52$ year and maturity time $T=1/12$ year. This model has been studied with least squares Monte Carlo method in \cite{broadie2015risk}. The portfolio consists of the following positions:
\begin{enumerate}
  \item Long one down-and-out put option with strike $K_1=101$ and barrier $H_1=91$.
  \item Long one down-and-out put option with strike $K_2=110$ and barrier $H_2=100$.
  \item Short one down-and-out put option with strike $K_3=114.5$ and barrier $H_3=104.5$.
\end{enumerate}
We aim to estimate the risk measure $\alpha=\mathbb{E}[(L_{\tau}-c)^+]$, where the threshold is the $95$-th percentile of the portfolio loss $L_{\tau}$, i.e. $c=\mathrm{VaR}_{0.95}(L_{\tau})=0.3608$. { Let $\underline{F}_T$ define the minimum asset price on $[\tau, T]$, and $F_T$ define the final asset price, then the portfolio loss at time $\tau$ is given by 
\be
L_{\tau}&=&e^{-r(T-\tau)}\mathbb{E}\left[(K_3-F_T)^+I(\underline{F}_T>H_3)-(K_2-F_T)^+I(\underline{F}_T>H_2)\right.\nonumber \\ &&\left.-(K_1-F_T)^+I(\underline{F}_T>H_1)|\mathcal{F}_{\tau}\right]-(P_3-P_1-P_2),\nonumber
\ee
where $I(\cdot)$ is the indicator function, and $P_1, P_2,P_3$ define the purchase prices of three options at time $\tau$, that is
$$P_i=e^{-r(T-\tau)} \mathbb{E}\left[\mathbb{E}\left[(K_i-F_T)^+I(\underline{F}_T>H_i)|\mathcal{F}^{\mathbb{Q}}_{\tau}\right]\right],\ \  i=1,2,3.$$
where $\mathcal{F}^{\mathbb{Q}}_{\tau}$ means the outer scenario generated by risk-free interest rate $r$.
}\end{example}

We compare efficiency and accuracy of three methods: (1) standard nested Monte Carlo simulation, (2) least squares Monte Carlo introduced in \cite{broadie2015risk}, and (3) sample recycling method proposed in this paper. Note that the payoff of a barrier option depends only on the minimum underlying asset price and the final asset price on time $[\tau,T]$. In the first numerical calculation, we shall simulate these two quantities instead of sampling the entire sample path (c.f. \cite{Martin2010}). Also note that the closed form expression for the portfolio losses $L_{\tau}$ given a risk factor scenario can be found in \cite{Haug2007}. Therefore, the risk measure $\alpha$ can be precisely computed by the simulation in the outer stage. Details of each method can be found below.

\begin{itemize}

  \item Standard nested Monte Carlo simulation

  It is known from \cite{broadie2015risk} that, given a fixed computing budget $k=mn$, the asymptotically optimal choice to minimize the MSE of the estimator is given by $n^*=\beta k^{2/3}$ outer stage scenarios and $m^*=k^{1/3}/\beta$ inner stage paths, where $\beta$ is determined by minimizing the asymptotic MSE and is difficult to derive. In this example, we use the optimized parameter value $\beta^{*}=0.076$ suggested by \cite{broadie2015risk}. In this numerical example, we set $k=10^6$ for the budget allocation, which results in $m^*=1,316$ inner paths and $n^*=760$ outer scenarios. Each scenario or path is based on the simulation of  $(\underline{F}_T, F_T)$, and the simulation method can be found in \cite{Haug2007}.
 \item Risk estimation via regression (least squares Monte Carlo)
 
The reference points are chosen equidistantly. We break the range of asset values from the $760$ scenarios into $10$ intervals of equal length. We select the right boundary points as sample outer scenarios and generate corresponding inner paths.
 The portfolio loss is evaluated under each outer scenario and the corresponding set of inner paths. Then we apply the method introduced in \cite{broadie2015risk} with the basis function set $\Phi^{(2)}$, including $1,\ F_{\tau},\ (F_{\tau}-H_1)^+,\ (F_{\tau}-H_2)^+,\ (F_{\tau}-H_3)^+,$ and their corresponding squared functions. We then use the approximate functional relationship by regression to determine portfolio loss under each of the $760$ outer scenarios. 

\item Sample recycling method:

Since we can simulate the minimum asset price $\underline{F}_T$ and the final asset price $F_T$ on $[\tau, T]$ directly, then the performance of portfolio on $[\tau,T]$ can be determined entirely by the pair $Z=(Z_1,Z_2)=(\underline{F}_T, F_T)$. The joint density function of $(\underline{F}_T, F_T)$ is already known (see for example \cite{Martin2010}) as follows,
$$p_i(z_1,z_2)=\frac{2}{\sqrt{2\pi}}\exp{\left(-\left(\widehat{z_2}-\frac{r-0.5\sigma^2}{\sigma_1}\sqrt{T-\tau}\right)^2/2\right)}\left(\widehat{z_2}-2\widehat{z_1}\right)\exp\left(-2\widehat{z_1}(\widehat{z_1}-\widehat{z_2})\right),$$
where $$\widehat{z_2}=\frac{\ln( z_2/x_i)}{\sigma_1\sqrt{(T-\tau)}}, \widehat{z_1}=\frac{\ln (z_1/x_i)}{\sigma_1\sqrt{(T-\tau)}}.$$
Then the likelihood can be calculated by \eqref{p_rat} as follows  $$
    p_{i|1}(z_1,z_2)=\exp\left( \ln \left(\frac{x_1}{x_i}\right) \frac{\ln \left(\sqrt{x_1x_i} z_2/z_1^2\right) +(r-0.5\sigma_1^2)(T-\tau)}{\sigma_1^2(T-\tau)}\right)
    \frac{\ln \left(x_iz_2/z_1^2\right)}{\ln \left(x_1z_2/z_1^2\right)},$$
     where $(x_k=F_{\tau}^{(k)}, k=1,2,\ldots,n)$ is an  i.i.d sample of $F_{\tau}$ and $x_1$ is the referred scenario.
    We use the same method to determine the $10$ referred outer scenarios as in the least squares Monte Carlo method. The evaluation of portfolio loss is carried out in the way outlined in Section \ref{sec:multi}.
\end{itemize}

\begin{table}[h]
\begin{center}
\begin{tabular}{c|c|r}
\hline
 Estimator& MSE & \  \ Time (secs) \\
\hline
Optimal standard nested estimator &$3.1980\times10^{-5}$&$2836.9837$\\
\hline
Sample recycling method &$5.3013\times10^{-5}$& $84.7044$\\
\hline
Regression &$8.4838\times10^{-5}$& $57.5254$  \\
\hline
\end{tabular}
\caption{MSE and run time for the barrier option portfolio with $1000$ independent trials} 
\label{bar_min}
\end{center}
\end{table}

Table \ref{bar_min} displays the MSE and run time for the above-mentioned methods over $1000$ independent trials. It shows that both sample recycling method and regression consume significantly less time than Monte Carlo approach with MSE of the same order. While the sample recycling method in this example requires more time than the least squares Monte Carlo but achieves higher accuracy. It should be pointed out that the determination of optimal parameter $\beta^\ast$ for the budget allocation method requires searching over a set of different potential values. These values are dependent on the specific form of risk measure under consideration and only known for a limited number of risk measures. It is often difficult to determine such values for general risk measures.

%




\begin{table}[h]
\begin{center}
\begin{tabular}{c|c|c}
\hline
 Number of reference points & MSE & \  \ Time (secs) \\
\hline
2 &$8.5734\times10^{-4}$&$74.1199$\\
\hline
5 &$3.8185\times10^{-5}$& $78.7044$\\
\hline
8&$3.2701\times10^{-5}$& $80.6645$  \\
\hline
\end{tabular}
\caption{MSE and run time for sample recycling method with various reference points}
\label{barrier_scr}
\end{center}
\end{table}

It should also be pointed out that the sample recycling method requires fewer reference points to approximate losses. Table \ref{barrier_scr} shows the MSE and time consumption for two, five and eight reference points. There are a total of $1,000$ independent trials for each case. In comparison, the regression method needs at least $10$ sample points because there are $9$ basis functions in this example.


\begin{table}[h]
\begin{center}
\begin{tabular}{c|c|c}
\hline
& $m\gamma$ &$m\delta$ \\
\hline
Time(Sec) &$2.2160\times10^{-3}$&$6.8562\times10^{-5}$\\
\hline
\end{tabular}
\caption{The values of $m\gamma$  and $m\delta$.}
\label{CE_barrier}
\end{center}
\end{table}
\begin{figure}[h]
\centering
 \includegraphics[width=0.7\textwidth]{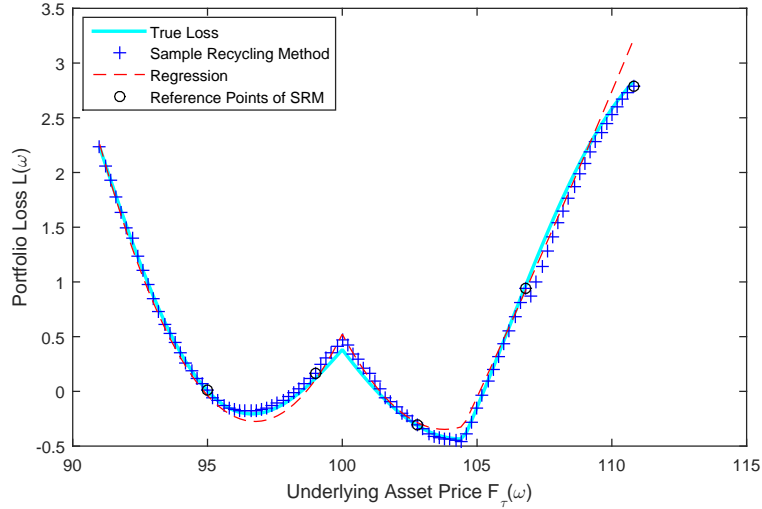}
\caption{Illustration of approximation of loss in barrier options.
}\label{barrier_f1}
\end{figure}

To tie it to the earlier discussion on computational effort, we show in Table \ref{CE_barrier}  computational efforts required for nested simulations. It is clear that in this case that, the simulation time of inner paths and the computation time for inner loop evaluation, $m\gamma$, is greater than the computation time of the likelihood, $m\delta$. In this example,  we use $5$ reference points with the sample recycling method. By definition \eqref{cesn} and \eqref{cesr}, we calculate computational efforts of $1,000$ independent trials given by
\be
&&\text{CE}_{\mathrm{SN}}=1000\times 1000\times m\gamma=2216.00346 \nb \\
&&\text{CE}_{\mathrm{SR}}=1000\times (5 \times m\gamma+995 \times m\delta)=11.4800+68.2193=79.6993.
\ee

Figure \ref{barrier_f1} shows a comparison of estimations for expected excess loss by regression and the sample recycling method. The light blue line represents the true value of expected excess loss $\alpha$ as a function of asset price $F_\tau$. The symbol $+$ shows estimates by the sample recycling method and the dashed line provides estimates by the regression. Both methods produce quite accurate estimates. The regression tends to overestimate for large asset values and underestimate in modest small asset values (between $96$ and $100$). In this graph, the regression approach is based on $20$ equidistant sample points, whereas the sample recycling method uses $5$ reference points, which are shown by the symbol $\circ.$ For the sample recycling methods, we break the range of asset prices into five blocks of equal lengths and use  the right-end point as the reference point for each block. One would notice that expected excess losses are either all overestimated or underestimated in each block.

\bigskip

\begin{example}\label{ex_asian}
Consider a portfolio of financial derivatives written on five underlying assets. Assume that the initial assets prices are all $F_0=100$, and that the assets share common real-world drifts of $\mu=8\%$ and annual volatility of $\sigma=20\%$. The risk-free continuously compounding interest rate is $r=3.5\%$. The asset price processes are assumed to be mutually independent. Suppose that the portfolio consists of $10$ short positions of at-the-money (average price) Asian call options on five underlying assets. All options have the same maturity date $T=1/12$ years and the portfolio is evaluated at $\tau=1/52$ years from now. We want to estimate the expected excess loss $\rho=\mathbb{E}[(L_{\tau}-c)^+]$ with threshold $c$ is the the $99$-th percentile of the portfolio loss, i.e. $c=\mathrm{VaR}_{0.99}=114.8151$. Let $F_{j,t}$, $j=1,2,\ldots,5$ represent the five underlying assets prices and $\Bar{F}_{j}$ represent the arithmetic price on $[\tau,T]$, then the portfolio loss can be given by 
$$L_{\tau}= \mathbb{E}\left[10 e^{-r(T-\tau)}\sum_{j=1}^5(\Bar{F}_{j}-100)^+|\mathcal{F}_{\tau}\right]-C$$
where $C$ is the purchase price of the portfolio, the price of Asian option can be approximated by built-in function of Matlab.
\end{example}
{ In this example, we use the built-in function \textbf{asianbylevy} of Matlab to approximate the closed form pricing solution of continuous arithmetic Asian options  \citep*{Levy1992}, which give rise to the true value of the loss of the portfolio.}
Detailed specification of each method is described below.
\begin{itemize}
  \item Nested Monte Carlo simulation:  In  this  numerical  example,  we  set $k= 10^6$ for the budget allocation and set  $n=1,000$ outer scenarios and $m=1,000$ inner paths to estimate the expected excess loss $\rho$. The portfolio loss is estimated  by simulating the entire sample path of $F_t$  as  Example 3.1. 
  \item Risk estimation via regression: 
  We choose basis functions up to fifth order polynomials. Specifically,  let $F_{j,t}$, $j=1,2,\ldots,5$ represent the five underlying assets prices, the basis functions contain all the following functions:
      $$1,\ F_{j,\tau},\ (F_{j,\tau})^2,\ (F_{j,\tau})^3,\ (F_{j,\tau})^4,\ (F_{j,\tau})^5,\ \ \ j=1,2,\ldots,5.$$
      We use $50$ simulated sample points (each with the inner path number $m=1000$) to perform the regression and to get the proxy function.  The loss $L_{\tau}$ on the  sample points is simulated by Monte Carlo, and the rest $950$ loss value is approximated by the proxy function. 
  \item Sample recycling method: To calculate the loss of the portfolio, we simulate the entire sample path of the underlying assets. Recall the discussion in Section \ref{sec_weight}, $p_{i|1}(\cdot)$ can be determined by the density function of $Z_j=F_{{j,\tau}+\Delta t},j=1,2,\cdots,5$ because of the Markov property.  For each underlying asset, the weight used in the evaluation of Asian options is same as Example  \ref{exa1}
      $$p_{i|1}(z_j)=\exp\left(\ln\left(\frac{x_i}{x_1}\right)\frac{-0.5\ln(x_1x_i)+(r-0.5\sigma^2)(\Delta t)+\ln(z_j)}{\sigma^2(\Delta t)}\right), \ \
   j=1,2,\ldots,5,$$
   where $\Delta t$ is the time step used to simulate the entire sample path of underlying assets. Finally, we take $\Delta t=1/624$ in this numerical example. We divide $10$ blocks for each underlying asset and choose the  intermediate point as the reference point. Then the total number of reference points is $10\times 5=50$.  
\end{itemize}
Table \ref{asian_1} shows that both sample recycling method and regression are more efficient than nested simulation method. The efficiency of sample recycling method is due to the computational effort $m \delta < m \gamma $ (see Table \ref{CE_asian} ). From \eqref{cesn} and \eqref{cesr}, we can calculate the follows results
\be
&&\text{CE}_{\mathrm{SN}}=5\times 1000 \times 1000 \times m\gamma=5682.7\nb\\
&&\text{CE}_{\mathrm{SR} }=5 \times 1000(10 \times m\gamma+990\times m \delta)=55.6827+290.8549=346.5376\nb
\ee
In Table \ref{asian_1}, the MSE of sample recycling method has the same magnitude with nested simulation method, but
the regression method often leads to wrong results because the sample points are insufficient for a basis set containing 26 basis functions. The accuracy of the regression method can be improved by increasing the number of sample points (see Table \ref{asian_reg_2}), but each sample points needs computational efforts $  m \gamma$ to simulate the value in once trial, then Table \ref{asian_reg_2} shows that increasing the number of sample also significantly increases the computational efforts. On the contrary, Table \ref{asian_srm} shows the MSE and run time when increasing the reference points for each Asian option from 15 to 30 in the sample recycling method,  and the results showed that the sample recycling method performs stably using different number of reference points, and the number of reference has less influence to the computational effort . 

\begin{table}[h]
\begin{center}
\begin{tabular}{c|l|r}
\hline
 Estimator& MSE & \  \ Time (secs) \\
\hline
Standard nested estimator &$5.9793\times10^{-3}$&$5724.0761$\\
\hline
Sample recycling method &$5.6534\times10^{-3}$& $379.3743$\\
\hline
Regression &$5.3235$& $315.9398$  \\

\hline
\end{tabular}
\caption{MSE and run time for Asian option portfolio with $1,000$ independent trials.}
\label{asian_1}
\end{center}
\end{table}
{\color{blue}
\begin{table}[h]
\begin{center}
\begin{tabular}{c|c|c}
\hline
& $m\gamma$ &$m\delta$ \\
\hline
Time (secs) &$1.113654\times10^{-3}$&$6.1233\times10^{-5}$\\
\hline
\end{tabular}
\caption{Comparison of computational efforts}
\label{CE_asian}
\end{center}
\end{table}

}
\begin{table}[h]
\begin{center}
\begin{tabular}{c|c|r}
\hline
 Number of reference points& MSE & \  \ Time (secs) \\
\hline
$15\times5$ &$6.5345\times10^{-3}$& $387.8045$\\
\hline
$20\times5$ &$5.9153\times10^{-3}$& $403.5701$\\
\hline
$25\times5$&$6.2073\times10^{-3}$& $416.9511$  \\
\hline
$30\times5$&$6.1332\times10^{-3}$& $422.7213$  \\
\hline
\end{tabular}
\caption{MSE and run time for sample recycling method with different reference points}
\label{asian_srm}
\end{center}
\end{table}
\begin{table}[h]
\begin{center}
\begin{tabular}{c|l|r}
\hline
 Number of reference points& MSE & \  \ Time (secs) \\
\hline
75 &$5.3095$&$403.114$\\
\hline
100 &$1.5625$& $580.9215$\\
\hline
125 &$5.7268\times10^{-2}$& $826.877$ \\
\hline
150&$3.8298\times10^{-2}$& $1022.6205$  \\

\hline
\end{tabular}
\caption{MSE and run time for regression with different sample points}
\label{asian_reg_2}
\end{center}
\end{table}

\begin{example}\label{ex_gmwb}
Another common application of nested Monte Carlo simulation is on the calculation of risk measure for variable annuity guaranteed benefits. Consider one of the most common investment guarantees on variable annuity products, known as the guaranteed minimum withdrawal benefit (GMWB). Suppose that the instantaneous change in fund value is the net effect of proportional return from equity-linking less percentage rider charges and fixed withdrawal
\be
dF_t&=&\frac{F_t}{S_t}d S_t-m_f F_t dt-w dt \nb \\
&=& ((r-m_f)F_t-w)dt+\sigma F_tdW_t, \nb
\ee
where $S_t$  is the equity-index driven by \eqref{geome}, $m_f > 0$ be the rate per time unit of total
fees charged by the insurer, and $w$ be the guaranteed rate of withdrawal per time unit. Let $G$ be the initial deposit, the GMWB rider provides safeguards to
the continuous withdrawal until the initial deposit is completely refunded, i.e. the GMWB matures
at time $T = G/w$. In this example, we take $F_0=G$, meaning that the policyholder is guaranteed to receive a full refund of his or her premium payments. It is only when the account value is depleted prior to the maturity $T$ that the maximum withdrawal rate $w$ is paid at the cost of the insurer. Therefore, the present value of the cost to an insurer of GMWB rider is given by 
$\int_{0}^Te^{-rs}wI(F_s\leq 0)ds$, where $I(\cdot)$ is an indicator function. On the other hand, the insurer receives the distribution of fees from the third party
fund manager, which are often a fixed percentage of the policyholder's account until the account value hits zero. Thus the accumulated present value of the fee income is given by $\int_{0}^Te^{-rs}m_fF_sI(F_s>0))ds$. Therefore,
the liability of insurer at time $\tau$ is given by
\be
L_{\tau}=\mathbb{E}\left[\left.\int_{\tau}^Te^{-r(s-\tau)}(wI(F_s\leq 0)-m_fF_sI(F_s>0))ds\right|\mathcal{F}_{\tau}\right].
\ee
 We calculate the risk measure  $\rho= \text{VaR} _{0.7}[L_{\tau}]$ by Monte Carlo and sample recycling method.
For the numerical calculation, we take $F_0=G=1, \mu=0.08, r=0.05, \sigma=0.2, w=0.1, m_f=0.01$. Suppose that the withdrawal benefit expires in $T=10$ years and the risk measure is evaluated in $\tau=5$ years.
\begin{itemize}
  \item Nested Monte Carlo simulation: We use $n=1000$ outer scenarios and $m=1000$ inner paths to estimate the risk measure. We estimate the liability by simulating the entire sample path of $F_t$ on $[\tau,T)$ with $\Delta t=0.05$, and risk factor $F_t$ is simulated by the following recursion
 \be\label{gmwb_sim}
 F_{h+1}=F_h\exp((r-m_f-0.5\sigma^2)\Delta t+\sigma \sqrt{\Delta t}X_{h+1})-w\Delta t, \  \  h=k,k+1,...,K.
 \ee
  where $X_1,...,X_K$ are independent draws from a standard normal distribution.
\item Risk estimation via regression:  We choose basis functions up to fifth order polynomials: $1,\ F_{\tau},\ (F_{\tau})^2,\ (F_{\tau})^3,\ (F_{\tau})^4,\ (F_{\tau})^5$. We use $50$  sample points (each with the inner path number m = 1000) to perform
the regression. The sample points are chosen the right endpoints equidistantly in each trial.
\item Sample recycling method: The liability is determined by the entire sample path of $ Z=F_t, t\geq \tau$.  From the recursion equation \eqref{gmwb_sim}, we have 
$$(F_{h+1}+\omega \Delta t)/F_h \sim N((r-m_f-\frac{\sigma^2}{2})\Delta t, \sigma^2\Delta t).$$  
 then we can calculate the weight as follows
 \be
 p_{i|1}(z)=\exp\Big(\frac{\ln(x_i/x_1)}{\sigma^2\Delta t}
\Big(-\frac{1}{2}\ln(x_1x_i)-(r-m_f-\frac{1}{2}\sigma^2)\Delta t+\ln (z+w \Delta t) \Big)\Big).\nb
 \ee
\end{itemize}
The first reference point $x_1^{ref}$ is the maximum value of the i.i.d samples  $x_k=:F^{(k)}_{\tau},k=1, 2\cdots, n$, and the rest of reference points are determined by  $x_{k}^{ref}=\inf\{x\in(x_1,x_2,\cdots,x_n)|x_{k-1}^{ref}/x<1.1\}$. In this example, the number of  reference points is $50$ or so. 
\end{example}
\begin{table}[h]
\begin{center}
\begin{tabular}{c|c|c|r}
\hline
 Estimator& VaR$_{0.7}$ &Stand. Dev. & \  \ Time (secs) \\
\hline
Standard nested estimator &$1.0050\times10^{-2}$&$5.6411\times10^{-3}$&$3890.5233$\\
\hline
Sample recycling method &$1.0152\times10^{-2}$&$6.0046\times10^{-3}$& $299.9947$\\
\hline
Regression
&$4.5300\times10^{-2}$&$9.8774\times10^{-3}$&$220.3021$  \\
\hline
\end{tabular}
\caption{The standard deviation and time for the GMWB example with 1000 independent trials.}
\label{gmwb}
\end{center}
\end{table}
\begin{table}[h]
\begin{center}
\begin{tabular}{c|c|c}
\hline
& $m\gamma$ &$m\delta$ \\
\hline
Time (secs) &$3.8349\times10^{-3}$&$5.9691\times10^{-5}$\\
\hline
\end{tabular}
\caption{The values of $m\gamma$  and $m\delta$.}
\label{CE}
\end{center}
\end{table}
 In Table \ref{gmwb}, we calculate the risk measure  $\text{VaR}_{0.7}$  by three methods. In this GMWB example, we cannot find the analytical solution, then we use the standard deviation (SD) of $1000$ trials to present the stability. Both nested Monte Carlo method and sample recycling method show high accuracy. However, the sample recycling method can be less time consuming than the  nested Monte Carlo method. Table \ref{CE} gives the values of $m\gamma$ and $m\delta$, the average time of 1000 independent trials. It is shown that the simulation and calculation process of $g(Z_{i,j})$ of standard nested Monte Carlo method is equal to 64 times of that of $p_{i|1}(\cdot)$. It is important to note that we ran 1000 independent trials and used 50 sample points in the sample recycling method. It follows from \eqref{cesn} and \eqref{cesr} that computational efforts are given by
\be
&&\text{CE}_{\mathrm{SN}}=1000\times 1000 \times m\gamma=3834.95 \nb \\
&&\text{CE}_{\mathrm{SR}}=1000(50\times m\gamma+950 \times m\delta)=191.745+56.70645=248.45.  \nb
\ee
It is clear that the main time consumption is from the simulation and calculation of $g(Z_{i,j})$.

\section{Non-parametric method} \label{sec:DR}

In practice, equity scenarios are typically generated from a sophisticated economic scenario generator. The underlying stochastic models are sometimes unknown to end users of equity scenarios. Therefore, it is possible that the likelihood (distorted weight) in  \eqref{inner1} is not known by analytical formula. In this section, we develop a non-parametric sample recycling method, which does not require prior knowledge about the underlying stochastic model. It is particularly useful when the likelihood cannot be derived explicitly or when underlying asset paths are generated by the empirical data rather than a specific model. .
\subsection{Likelihood ratio estimation}

In this section, we introduce a naive estimation method for the likelihood ratio function $p_{i|1}(\cdot)$. Despite its simplicity, this method demonstrates high
accuracy for loss estimation by numerical examples.

\begin{algorithm}[H] 
\caption{Estimate likelihood ratio $p_{i|1}(\cdot)$ using $\hat{p}_{i|1}(\cdot)$ } 
\label{algo:lambda_est} 
\begin{algorithmic}
 \State Generate $n$ outer scenarios
 \State Conditioned on $i$-th ( $i=1, \cdots, n$) outer scenario, generate $m$ i.i.d. sample points $\{F_{k+1}^{(i,j)}\}_{j=1}^m$, 
 \State Sort the  sample points in increasing order $\{F_{k+1}^{(i,[j])}\}_{j=1}^m$ 
 \State Separate the reference sample points $\{F_{k+1}^{(1,[j])}\}_{j=1}^m$  into $l$ sets and find the counts of each set $(n_1^{(1)}, n_2^{(1)},\cdots, n_l^{(1)})$
 \For {$\ i = 1 \  \text{to}  \ n$}
 \State Find the counts of  target sample points in each set $(n_1^{(i)}, n_2^{(i)},\cdots, n_l^{(i)})$ 
\For {$\ j = 1 \  \text{to}  \ m$}
 \If {$y_{k+1}=F_{k+1}^{(i,[j])}$ in $a$-th set}
\State $\hat{p}_{i|1}(y_{k+1})=\frac{n_a^{(i)}}{n_a^{(1)}}$
\EndIf
\EndFor
\State $p_{i|1}(\cdot)\leftarrow\hat{p}_{i|1}(\cdot)$
\EndFor
 \end{algorithmic}
\end{algorithm}

To illustrate this method, we only consider one risk factor $\{F_t\}_{t\geq 0}$ as Section 3.1 and set an independent and identically distributed inner loop sample $\{F_{k+1}^{(i, j)}\}^m_{j=1}$ for $i=1, \cdots, n,$ generated from outer scenario $F^{(i)}_k$ for a univariate Markov stochastic model. Note that, the set $\{F_{k+1}^{(1, j)}\}^m_{j=1}$ shall be used as a reference point, while others are considered as target points. The method can be broken down into the following steps.

\begin{enumerate}
\item Sort the data set for the sample point in an increasing order: $\{F_{k+1}^{(i, [j])}\}_{j=1}^{m}$ where $[j]$ indicates the $j$-th order statistic. In other words, $F_{k+1}^{(i, [1])}\le F_{k+1}^{(i, [2])} \le \cdots \le F_{k+1}^{(i, [m])}.$
\item Seperate the set of integers $\{1, \cdots, m\}$ into $l$ sets with break points $m_0=0, m_1, m_2, \cdots, m_l=m$. Denote the $a$-th interval of risk factor by
$$ I_a := \big( F_{k+1}^{(1, [m_{a-1}])}, F_{k+1}^{(1,[m_a])} \big]. $$
\item Count the number of observations of samples from target scenario and reference scenario in each interval respectively. Denote the counts by $ ( n_1^{(i)}, n_2^{(i)}, \ldots, n_l^{(i)} ) $ and $ ( n_1^{(1)}, n_2^{(1)}, \ldots, n_l^{(1)} ) $.
\end{enumerate}
We can construct the following likelihood ratio estimator 
\be\label{eq:density_ratio_est}
\widehat{p}_{i|1}(y_{k+1}) =  \frac{n_a^{(i)}}{n_a^{(1)}} \ \ \text{for}\ \  y_{k+1}\in I_a, \ee
resulting in an empirical likelihood ratio function. 
If we choose intervals based on quantiles, we could construct the intervals such that there are equal numbers of points in each interval, i.e. $n_a^{(1)} = m/l$ for each $a$, further simplifying the estimator to
$$ \widehat{p}_{i|1}(y_{k+1}) =  \frac{l \cdot n_a^{(i)}}{m}\ \ \text{for}\ \  y_{k+1}\in I_a. $$
\begin{proposition}
Given i.i.d sample $\{F_{k+1}^{(i, j)}\}_{j=1}^{m}$ for $i=1, \cdots, n$, the estimator \eqref{eq:density_ratio_est} converges pointwise to its true value for each valid input $ y_{k+1} \in I_a$:
$$ \widehat{p}_{i|1}(y_{k+1}) =  \frac{n_a^{(i)}}{n_a^{(1)}} \rightarrow p_{i|1} (y_{k+1}) \quad \text{as} \quad m \rightarrow + \infty,\ l \rightarrow + \infty. 
$$
\end{proposition}
\begin{proof}
By Glivenko-Cantelli Theorem, an empirical distribution function uniformly converges to the true cumulative density function as the number of i.i.d observations approaches infinity. Let $Q_i$ denote the cumulative distribution function for $F_{k+1}$ under measure $\mathbb{Q}_i$ and $\widehat{Q}_i$ be the empirical distribution function. Then for each scenario $i$,
$$ || \widehat{Q}_i - Q_i  ||_{\infty} \rightarrow 0. $$
Notice that \eqref{eq:density_ratio_est} can be written as
$$ \widehat{p}_{i|1}(F_{k+1}) =  \frac{n_a^{(i)}}{n_a^{(1)}} = \frac{n_a^{(i)} / m}{n_a^{(1)} / m} = \frac{\widehat{Q}_i(F_{k+1}^{(1, [m_{a}])}) - \widehat{Q}_i(F_{k+1}^{(1,[m_{a-1}])})}{\widehat{Q}_1(F_{k+1}^{(1,[m_{a}])}) - \widehat{Q}_1(F_{k+1}^{(1,[m_{a-1}])})}, $$
As $m, l \rightarrow \infty$, we obtain
\begin{align*}
\frac{\widehat{Q}_i(F_{k+1}^{(1, [m_{a}])}) - \widehat{Q}_i(F_{k+1}^{(1,[m_{a-1}]})}{F_{k+1}^{(1, [m_{a}]}-F_{k+1}^{(1, [m_{a-1}]}} &\rightarrow p_i(F_{k+1}).
\end{align*}
 Therefore, their ratio approaches the ratio of limit at each given $F_{k+1}$, 
$$ \widehat{p}_{i|1}(F_{k+1}) \rightarrow \frac{p_i(F_{k+1})}{p_1(F_{k+1})} = p_{i|1}(F_{k+1}). $$
\end{proof}

To illustrate the estimation, we choose two outer scenarios $F_{\tau}=99$ (reference) and $F_{\tau}=99.2$ (target) with $\tau=1/52$ in the geometric Brownian motion, and we generate $1000$ sample points for each scenario.  The histogram and the empirical likelihood ratio function  demonstrate the result of Algorithm \ref{algo:lambda_est} in Figure \ref{nopara}. 
Five intervals are constructed based on the $20$-th, $40$-th, $60$-th, $80$-th quantiles of reference sample points, which means $l=5$. In this example, the left figure shows that $(n_1^{(i)}, \cdots, n_5^{(i)})=(122,174,184,243,272)$ and $n_k^{(1)}=200$ for $k=1,\cdots, 5$. We overlay the histograms and theoretical density functions for both reference point (blue) and target point (red) to show their differences in the left plot while the empirical likelihood ratio function (solid-line) and the true theoretical likelihood ratio (dashed-line) are shown on the right. Keep in mind that the true likelihood ratio function is typically not known in advance. The graph shows a reasonable estimate from empirical data.
\begin{figure}[h]
\centering
 \includegraphics[width=1.1\textwidth]{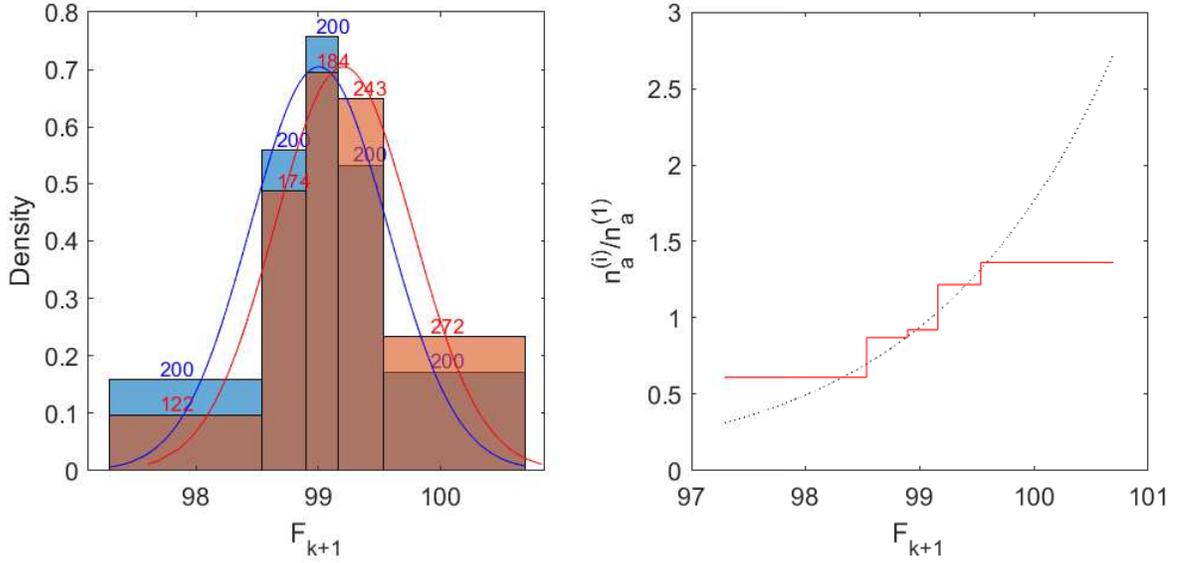}
\caption{The calculation of weight by nonparametric method}\label{nopara}
\end{figure}


Note that estimating probability density is a common question in machine learning. While this paper only discusses a naive method, we believe that many other methods can be used to estimate the likelihood, such as least square importance fitting\citep*{Kanamori2009}, kernel mean matching\citep*{Huang2007}, Kullback-Leibler importance estimation procedure and so on. \cite{Sugiyama2012} offers detailed accounts of machine learning methods.

\subsection{Non-parametric sample recycling method}
In the non-parametric setting, we estimate the theoretical likelihood ratio $p_{i|1}$ by an estimated $\widehat{p}_{i|1}()$ in the estimator \eqref{inner1}. Therefore, We obtain the \textit{empirical sample recycling estimator} of $L_i$,
\be\label{eq:SRM_emp}
\tilde{\tilde{L}}_i=\frac{1}{m}\sum_{j=1}^{m}\widehat{p}_{i|1}(Z_{1,j}) g(Z_{1,j}),
\ee
and the non-parametric sample recycling estimator of the risk measure $\rho$ is given by
\[\tilde{\rho}_{\mathrm{NSR}}=\frac1n \sum^n_{i=1} f(\tilde{\tilde{L}}_i).\]
Note that there is an additional source of randomness in this estimator --- likelihood ratio estimation. As the estimate requires no information about the underlying stochastic model, we do not expect this estimator to outperform the sample recycling method in the previous section. Nevertheless, the estimator \eqref{eq:SRM_emp} offers an appealing non-parametric framework when equipped with a reasonably fast and accurate algorithm to estimate likelihood ratios.

\begin{figure}[h]
\centering
 \includegraphics[width=1.1\textwidth]{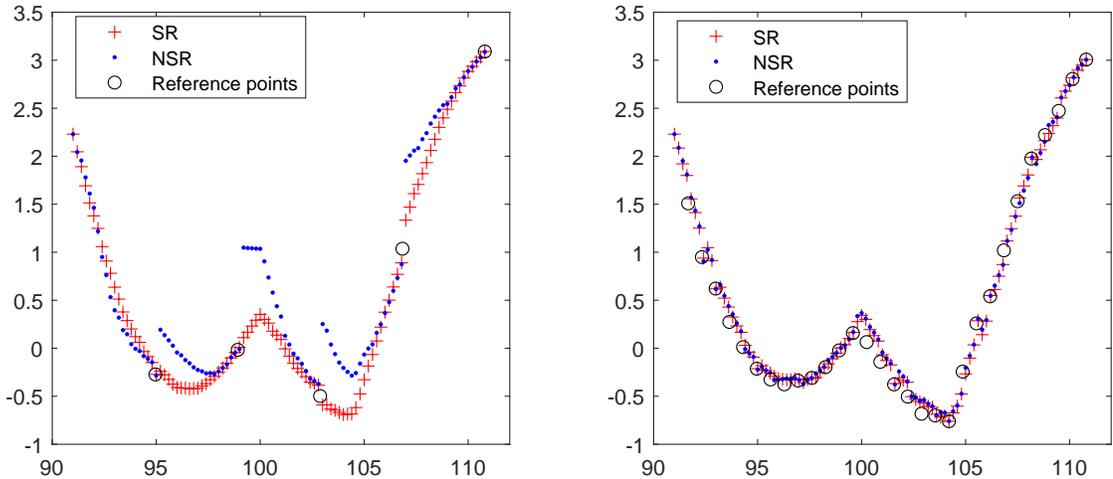}
\caption{Non-parametric sample recycling estimation of portfolio loss in Example \ref{ex_bar}
}\label{no_bar}
\end{figure}
\begin{figure}[h]
\centering
 \includegraphics[width=1\textwidth]{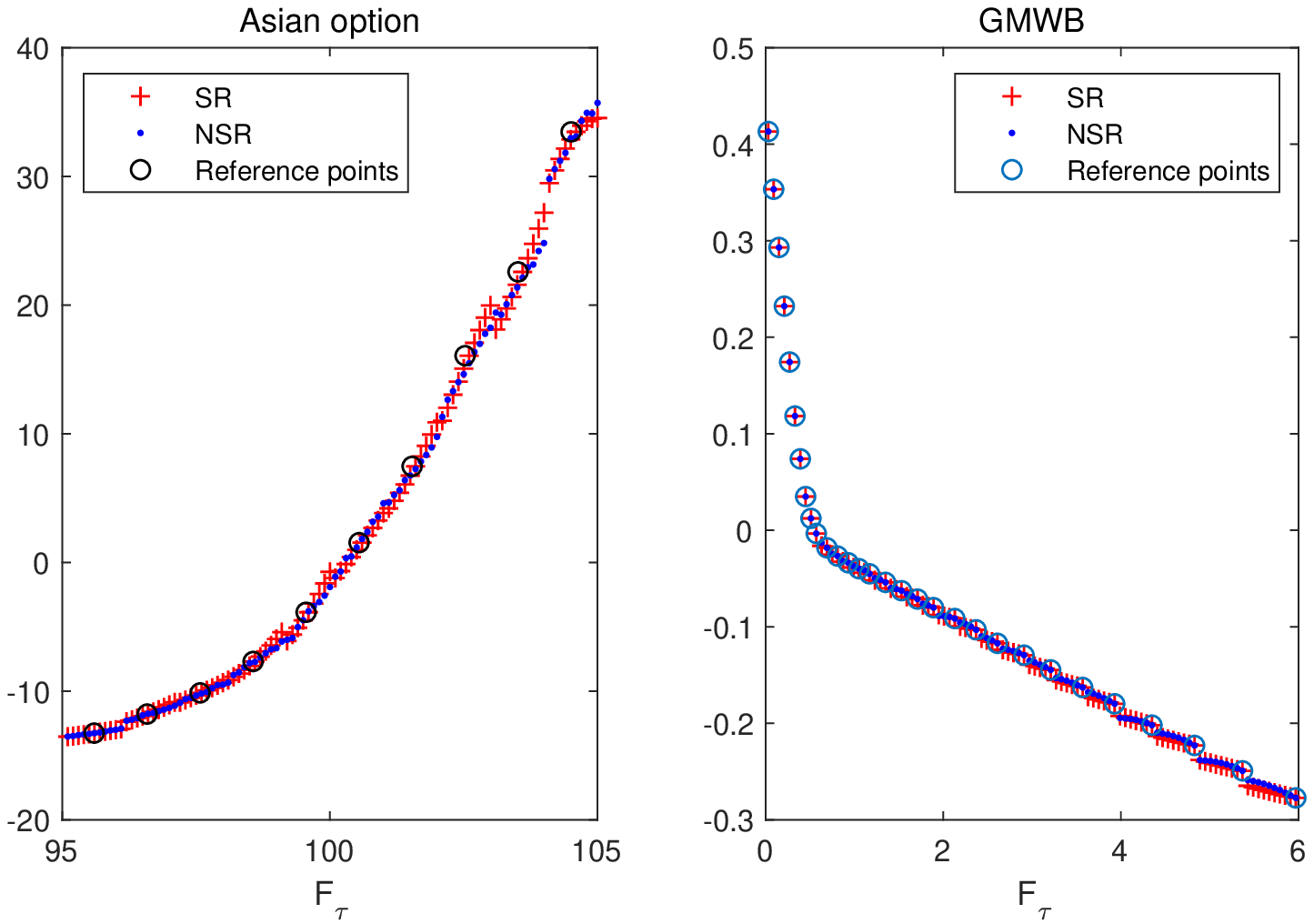}
\caption{Non-parametric sample recycling estimations of portfolio loss in Examples \ref{ex_asian} and \ref{ex_gmwb}
}\label{no_asi}
\end{figure}

To test the accuracy and efficiency of this non-parametric method, we re-run inner loop estimations in Examples \ref{ex_bar}--\ref{ex_gmwb} and compare results from the non-parametric sample recycling (NSR) method with those by the sample recycling (SR) method in \eqref{inner1}. For Examples \ref{ex_bar}, we consider losses of the barrier option portfolio for $100$ equidistant points in the range of equity price $[91,110.8]$. For both SR and NSR methods, we choose the same set of reference points to estimate the corresponding $L_{\tau}$ of other points in each example. We always use $l=5$ intervals for counting observations to estimate the corresponding likelihood ratios in all examples.
Figure \ref{no_bar} compares estimations of $L_i$ by both methods for the barrier option portfolio. The right endpoint is chosen as the reference point in each interval. It is clear from the left panel of Figure \ref{no_bar} that the SR method leads to fairly accuracy results even with only $5$ reference points and that the NSR method produce results with larger estimation errors. However, as we increase the number of reference points to ten, we observe from the right panel of Figure \ref{no_bar} that results from the NSR method are much closer to those from the SR method and hence improve significantly. We can apply the same technique to the other two examples. Figure \ref{no_asi} shows the comparison of results by both methods for the Asian option portfolio and the GMWB liability. In the estimation of portfolio loss in the basket of Asian options, we use $10$ reference points to estimate losses on $100$ equidistant equity values over the range $[95.1,105]$  and midpoints as the reference points for all intervals. In the estimation of the GMWB liability, we use $30$ reference points to estimate the GMWB liability for $100$ equidistant equity values over the range $[0.03,6]$. However, for this example we take a different approach to choose reference points. The reference points are chosen by right points in intervals of length determined by a geometric series. We first set the first reference point $x_1=6,$ and the rest of reference points are given by $x_{k}=\inf\{x\in(x_1,x_2,\cdots,x_n)|x_{k-1}/x<1.1\}$. The right panel of Figure \ref{no_asi} shows estimated GMWB liabilities based on $100$ equidistant points of asset prices. The graph clearly shows that both methods produce very similar results.



 \begin{table}[h]
 \begin{center}
 \begin{tabular}{c|c|c}
 \hline
  Number of reference points& MSE & \  \ Time(Sec) \\
 \hline
$30$ &$2.9717\times10^{-5}$& $ 92.0662 $\\
 \hline
 \end{tabular}
 \caption{MSE and time of $\tilde{\rho}_{\text{NSR}}$ for Example \ref{ex_bar}.}
 \label{bar_non}
 \end{center}
 \end{table}
 \begin{table}[h]
 \begin{center}
 \begin{tabular}{c|c|c}
\hline
  Number of reference points& MSE & \  \ Time (secs) \\
 \hline
 $5\times5$ &$6.5915\times10^{-3}$& $ 283.6537 $\\
 \hline
\end{tabular}
 \caption{MSE and time of $\tilde{\rho}_{\text{NSR}}$ for Example \ref{ex_asian}.}
 \label{asian_non}
 \end{center}
 \end{table}
 \begin{table}[h]
 \begin{center}
\begin{tabular}{c|c|c}
 \hline
$\text{VaR}_{}$& Stand Dev & \  \ Time (secs) \\
 \hline
$0.83757\times 10^{-2} $ &$6.1827\times10^{-3}$& $  368.8039$\\
 \hline
 \end{tabular}
 \caption{Standard deviation and run-time of $\tilde{\rho}_{\text{NSR}}$ for Example \ref{ex_gmwb}.}
 \label{gmwb_non}
 \end{center}
 \end{table}
 
 To further illustrate the implement of the NSR method, we extend these numerical examples further to show the computation of risk measures $\rho$ by the non-parametric estimator $\tilde{\rho}_{\text{NSR}}$. 
Comparing Table \ref{bar_non} with Table \ref{barrier_scr}, we observe results by both the non-parametric estimator $\tilde{\rho}_{\text{NSR}}$ and the original estimator $\tilde{\rho}_{\text{SR}}$. We use $30$ reference points in the inner estimation of non-parametric method to guarantee the accuracy. Table \ref{bar_non} indicates that it takes more time than sample recycling method because of the increased number of reference points. Nonetheless, the NSR method still outperforms the standard nested Monte Carlo.
Table \ref{asian_non} is the analogue of Table \ref{asian_1} for the non-parametric method. We use $5$ reference points to estimate each Asian option and hence the total number of reference points is $5 \times 5$. 
Table \ref{gmwb_non} corresponds to Table \ref{gmwb} with the non-parametric method. 
Both of these examples show that the non-parametric method has higher efficiency and accuracy than standard Monte Carlo.


As shown in previous numerical examples, the non-parametric sample recycling method is easy to implement. While it does not achieve the same level of accuracy as the original sample recycling method given a fixed set of reference points, one may have to resort to the non-parametric approach as the underlying model is unknown. The examples provide evidence to show that the non-parametric approach is a viable alternative whose accuracy improves with the size of reference points.

\section{Conclusion}

Most of existing techniques to reduce run-time for nested simulation are based on the replacement of inner loop simulations with curve fitting. The essence of these techniques is to develop a functional relationship between risk factors (equity values, interest rates, etc) and target features (insurance liability, Greek values) of inner loop calculations. Such a functional relationship can be approximated by multivariate interpolation or smoothing techniques such as least squares Monte Carlo. Nonetheless, these techniques often require a large size of economic scenarios to develop accurate enough functional relationships, which could also be costly to begin with. This paper proposes a new approach based on an entirely different strategy, which is to avoid approximate functional relationship and instead to save time by reducing repeated re-sampling of economic scenarios. The technique is to generate sample of risk factors under a small set of probability measures and recycle them by twisting likelihood ratios under other probability measures. The advantage of this approach is to reduce the number of sample generation for risk factors and subsequent inner loop evaluations. The disadvantage of such an approach is that the reduction of computational burden is achieved at the expense of increased sampling errors. This method is particularly suitable for long term products that require heavy computation for inner loop evaluation. 

While we have shown analytical solutions to distorted weights for various parametric models, we also consider the application of non-parametric sample recycling method to settings where the underlying model is either unknown or too complicated. The non-parametric is shown to be able to reproduce results, free of any information about the underlying model. It is less accurate than the sample recycling method but can be improved with an increased number of reference points. We only present a naive version of non-parametric likelihood ratio estimation as a proof-of-concept. However, there is a rich body of literature on machine learning techniques that can be used to estimate density ratios. Future work is needed to improve the naive method with more sophisticated machine learning for better accuracy and efficiency. 

\appendix

\section{Calculations} \label{sec:app}

\begin{example}\label{exam_variance}
In \eqref{riskmeas} and \eqref{inner}, we assume that $ X$ is given by a uniform random variable on $[-1,1]$ and $Z$ is a standard normal random variable. Consider expected value of loss where the loss is determined by $L=\mathbb{E}[g(Z)|X]=\mathbb{E}[\sqrt{2/\pi}\exp (-2(Z-X))|X].$ It follows from Proposition \ref{prop_variance} that
$$\text{Var}(\widehat{\rho}_{\mathrm{SN}})=\frac{1}{nm}\left(\frac{1}{\sqrt{\pi}}\left(\Phi(\sqrt{8/9})-0.5\right)-(\Phi(2/\sqrt{5})-0.5)^2\right)$$
and
\be
\text{Var}(\widetilde{\rho}_{\mathrm{SR}})= \frac{1}{mn^2}[B_2+(n-1)A_2+2(n-1)C+(n^2-3n+2)D-(B_1+(n-1)A_1)^2],\nb
\ee
where
\be
&&A_l=\int_{-1}^{1}\int_{-1}^{1}\left(\sqrt{\frac{2}{\pi}}\right)^{l}\frac{1}{4\sqrt{4l+1}} \exp\left(-\frac{1}{8l+2}\left((3l^2+l)x_i^2+(3l-13l^2)x_1^2+10l^2x_1x_i\right) \right)  dx_1 dx_i \nb \\
&&B_l=\left(\sqrt{\frac{2}{\pi}}\right)^{l}\sqrt{\frac{\pi}{2l}}\left[\Phi(\sqrt{4l/4l+1})-0.5\right], l=1,2\  \nb \\
&&C=\int_{-1}^{1}\int_{-1}^{1}\frac{1}{6\pi} \exp\left(-\frac{4}{9}x_i^2+x_1^2-x_1x_i \right)
 dz dx_1 dx_i.\nb\\
&&D=\int_{-1}^{1} \int_{-1}^{1}\int_{-1}^{1}\frac{1}{12\pi}\exp\left(-\frac{4}{9}(x_i^2+x_j^2)-x_1(x_i+x_j)+\frac{3}{2}x_1^2+\frac{x_ix_j}{9} \right)
 dx_1 dx_i dx_j.\nb
\ee
\end{example}
Consider an independent and identically distributed sample of $X$ denoted by $(X_1,X_2,...X_n)$. For each given $X_i$, we have $(Z-X_i)$ follows a normal distribution with mean $-X_i$ and variance $1$. We can therefore determine the coefficients.
$$p_{i|1}(Z)=\frac{\phi(Z+X_i)}{\phi(Z+X_1)}.$$
The follows gives the calculations of $A_l, B_l, C, D$.
\be
B_l&=&\int_{-1}^{1}\frac12\int_{-\infty}^{+\infty}\left(\sqrt{\frac{2}{\pi}}\right)^{l}\frac{1}{\sqrt{2\pi}}\exp\left(-2l(x-z)^2-\frac{z^2}{2}\right)dz dx\nb \\
 &=&\int_{-1}^{1}\frac12\int_{-\infty}^{+\infty}\left(\sqrt{\frac{2}{\pi}}\right)^{l}\frac{1}{\sqrt{2\pi}}\exp\left(-\frac{1}{2}\left(\sqrt{4l+1}z-\frac{4l}{\sqrt{4l+1}}x\right)^2-\frac{2l}{4l+1}x^2\right)dz dx \nb \\
 &=&\int_{-1}^{1}\left(\sqrt{\frac{2}{\pi}}\right)^{l}\frac{1}{2\sqrt{4l+1}}\exp\left(-\frac{2l}{4l+1}x^2\right)dx\nb \nb \\
 &=& \left(\sqrt{\frac{2}{\pi}}\right)^{l} \sqrt{\frac{\pi}{2l}} \int_{-\sqrt{4l/4l+1}}^{\sqrt{4l/4l+1}}\frac{1}{\sqrt{2\pi}}\exp\left(-\frac{x^2}{2}\right)dx\nb \\
 &=& \frac12\left(\sqrt{\frac{2}{\pi}}\right)^{l}\sqrt{\frac{\pi}{2l}} \left[\Phi(\sqrt{4l/4l+1})-\Phi(-\sqrt{4l/4l+1})\right] \nb \\
 &=&\frac12\left(\sqrt{\frac{2}{\pi}}\right)^{l}\sqrt{\frac{\pi}{2l}}\left[2\Phi(\sqrt{4l/4l+1})-1\right],\nb
\ee
where $$B_1=\Phi(2/\sqrt{5})-0.5\ \  \text{and} \ \ B_2=\frac{1}{\sqrt{\pi}}\left(\Phi(\sqrt{8/9})-0.5\right).$$
\be
A_l&=&\int_{-1}^{1}\frac{1}{2}\int_{-1}^{1}\frac{1}{2}\int_{-\infty}^{+\infty} \left(\sqrt{\frac{2}{\pi}}\right)^{l}\frac{1}{\sqrt{2\pi}} \exp\left(-\frac{1}{2}\left(l(z+x_i)^2-l(z+x_1)^2+z^2\right)-2l(z-x_1)^2\right) dz dx_1 dx_i \nb\\
&=& \int_{-1}^{1}\frac{1}{2}\int_{-1}^{1}\frac{1}{2}\int_{-\infty}^{+\infty} \left(\sqrt{\frac{2}{\pi}}\right)^{l}\frac{1}{\sqrt{2\pi}} \exp\left(-\frac{1}{2}\left(\sqrt{4l+1}z+\frac{n}{\sqrt{4l+1}}x_i-\frac{5l}{\sqrt{4l+1}}x_1\right)^2\cdot\right.\nb \\
 &&\left.\exp\left(-\frac{1}{8l+2}\left((3l^2+l)x_i^2+(3l-13l^2)x_1^2+5l^2x_1x_i\right) \right) \right)
 dz dx_1 dx_i \nb \\
&=& \int_{-1}^{1}\frac{1}{2}\int_{-1}^{1}\frac{1}{2}\left(\sqrt{\frac{2}{\pi}}\right)^{l}\frac{1}{\sqrt{4l+1}} \exp\left(-\frac{1}{8l+2}\left((3l^2+l)x_i^2+(3l-13l^2)x_1^2+10l^2x_1x_i\right) \right)  dx_1 dx_i \nb
\ee
where
\be
A_1= \int_{-1}^{1}\int_{-1}^{1}\frac{1}{4}\sqrt{\frac{2}{5\pi} } \exp\left(-\frac{2}{5}x_i^2+x_1^2-x_1x_i\right) dx_1 dx_i\nb
\ee
\be
 A_2= \int_{-1}^{1}\int_{-1}^{1}\frac{1}{6\pi} \exp\left(-\frac{1}{9}(7x_i^2-23x_1^2+20x_1x_i)\right) dx_1 dx_i.\nb
\ee

\be
C&=&\int_{-1}^{1}\frac{1}{2}\int_{-1}^{1}\frac{1}{2}\int_{-\infty}^{+\infty} \frac{2}{\pi}\frac{1}{\sqrt{2\pi}} \exp\left(-\frac{1}{2}\left((z+x_i)^2-(z+x_1)^2+z^2\right)-4(z-x_1)^2\right) dz dx_1 dx_i \nb\\
\mathbb{}&=& \int_{-1}^{1}\frac{1}{2}\int_{-1}^{1}\frac{1}{2}\int_{-\infty}^{+\infty} \frac{2}{\pi}\frac{1}{\sqrt{2\pi}} \exp\left(-\frac{1}{2}\left(3z+\frac{1}{3}x_i-3x_1\right)^2\exp\left(-\frac{4}{9}x_i^2+x_1^2-x_1x_i \right) \right)
 dz dx_1 dx_i \nb \\
 &=& \int_{-1}^{1}\frac{1}{2}\int_{-1}^{1}\frac{1}{3\pi} \exp\left(-\frac{4}{9}x_i^2+x_1^2-x_1x_i \right)
 dz dx_1 dx_i.\nb
\ee
\be
D&=&\int_{-1}^{1}\int_{-1}^{1}\int_{-1}^{1}\frac{1}{8}\int_{-\infty}^{+\infty} \frac{2}{\pi}\frac{1}{\sqrt{2\pi}} \exp\left(-\frac{1}{2}\left((z+x_i)^2+(z+x_j)^2-2(z+x_1)^2+z^2\right)-4(z-x_1)^2\right) dz dx_1 dx_idx_j \nb\\
 &=& \int_{-1}^{1} \int_{-1}^{1}\int_{-1}^{1}\frac{1}{8}\int_{-\infty}^{+\infty} \frac{2}{\pi}\frac{1}{\sqrt{2\pi}}  \exp\left(-\frac{1}{2}\left(3z+\frac{1}{3}(x_i+x_j)-3x_1\right)^2\right) \nb \\
 &&\exp\left(-\frac{4}{9}(x_i^2+x_j^2)-x_1(x_i+x_j)+\frac{3}{2}x_1^2+\frac{x_ix_j}{9} \right)
 dz dx_1 dx_i dx_j.\nb\\
 &=&\int_{-1}^{1} \int_{-1}^{1}\int_{-1}^{1}\frac{1}{12\pi}\exp\left(-\frac{4}{9}(x_i^2+x_j^2)-x_1(x_i+x_j)+\frac{3}{2}x_1^2+\frac{x_ix_j}{9} \right)
 dx_1 dx_i dx_j.\nb
\ee

{\bf Acknowledgments}

This work was supported by the Natural Science Foundation of Jiangsu Province [BK20200833]; the MOE Project of Humanities and Social Sciences [19YJCZH083]; National Natural Science Foundation of China [12001267].

\newpage
\bibliography{testbible}
\end{document}